\documentclass[orivec,envcountsect,envcountsame,runningheads,a4paper]{llncs}
\usepackage{cite}
\usepackage{amsmath,amsfonts,amssymb}
\usepackage{stmaryrd}
\SetSymbolFont{stmry}{bold}{U}{stmry}{m}{n}
\usepackage{scalerel}
\usepackage{proof}
\usepackage{needspace}
\usepackage{tikz}
\usepackage{float}
\floatstyle{boxed}
\restylefloat{table}
\restylefloat{figure}
\allowdisplaybreaks

\newcommand\letcase[3]{\ensuremath{\mathsf{letcase}\ #1 = #2\ \mathsf{in}\ \{#3\}}}
\newcommand\qletcase[3]{\ensuremath{\mathsf{letcase}^\circ\ #1 = #2\ \mathsf{in}\ \{#3\}}}
\newcommand\qvalue[1][]{value{#1}$^\circ$}
\newcommand\tr{\ensuremath{\mathsf{tr}}}
\newcommand\topr[1][1]{\ensuremath{\longrightarrow_{#1}}}
\newcommand\tsem[1]{\ensuremath{\llbracket #1\rrbracket}}
\newcommand\sem[2][\theta]{\ensuremath{\tsem{#2}_{#1}}}
\newcommand\fsem[2][\theta]{\ensuremath{\llparenthesis #2\rrparenthesis_{#1}}}
\newcommand\conv[1]{\ensuremath{\left[{#1}\right]}}
\newcommand\ket[1]{\ensuremath{|#1\rangle}}
\newcommand\bra[1]{\ensuremath{\langle #1|}}

\newcommand\had{\ensuremath{\mathsf{H}}}
\newcommand\cnot{\ensuremath{\mathsf{Cnot}}}
\newcommand\X{\ensuremath{\mathsf{X}}}
\newcommand\Z{\ensuremath{\mathsf{Z}}}
\newcommand\lambdens{\ensuremath{\lambda_\rho}}
\newcommand\qlambdens{\ensuremath{\lambda_\rho^\circ}}
\newcommand\D[1][n]{\ensuremath{\mathcal D_{#1}}}
\newcommand\emptybit{\ensuremath{\varepsilon}}

\newcommand\ext[1]{\ensuremath{\overline{#1}}}

\newcommand\trd[1]{\ensuremath{\mathsf{trd}({#1})}}
\newcommand\w[1]{\ensuremath{\mathsf{w}{(#1)}}}
\newcommand\para[1][e]{\ensuremath{(\mathsf b,\mathsf{#1})}}
\newcommand\toD{\ensuremath{\rightsquigarrow}}

\begin{document}

\title{A Lambda Calculus for Density Matrices\\
  with Classical and Probabilistic Controls} \titlerunning{A Lambda Calculus for
  Density Matrices} \author{Alejandro D\'{\i}az-Caro\thanks{Supported by
    projects STIC-AmSud 16STIC05 FoQCoSS, PICT 2015-1208 and the Laboratoire
    International Associ\'e ``INFINIS''.}} \authorrunning{A.~D\'{\i}az-Caro}

\institute
{Universidad Nacional de Quilmes \& CONICET\\
  Roque S\'aenz Pe\~na 352, B1876BXD Bernal, Buenos Aires, Argentina\\
  \email{alejandro.diaz-caro@unq.edu.ar} }

\maketitle 

\begin{abstract}
  In this paper we present two flavors of a quantum extension to the lambda
  calculus. The first one, $\lambdens$, follows the approach of classical
  control/quantum data, where the quantum data is represented by density
  matrices. We provide an interpretation for programs as density matrices and
  functions upon them. The second one, $\qlambdens$, takes advantage of the
  density matrices presentation in order to follow the mixed trace of programs
  in a kind of generalised density matrix. Such a control can be seen as a
  weaker form of the quantum control and data approach. 
\end{abstract}

\section{Introduction}
In the last decade several quantum extensions to lambda calculus have been
investigated,
e.g.~\cite{vanTonderSIAM04,SelingerValironMSCS06,PaganiSelingerValironPOPL14,ZorziMSCS16,ArrighiDowekLMCS17,ArrighiDiazcaroValironIC17,DiazcaroDowek17}.
In all of those approaches, the language chosen to represent the quantum state
are vectors in a Hilbert space. However, an alternative formulation of quantum
mechanics can be made using density matrices. Density matrices provide a way to
describe a quantum system in which the state is not fully known. More precisely,
density matrices describe quantum systems in a mixed state, that is, a
statistical set of several quantum states. All the postulates of quantum
mechanics can be described in such a formalism, and hence, also quantum
computing can be done using density matrices.

The first postulate states that a quantum system can be fully described by a
density matrix $\rho$, which is a positive operator with trace (\tr) one. If a
system is in state $\rho_i$ with probability $p_i$, then the density matrix of
the system is $\sum_ip_i\rho_i$. The second postulate states that the evolution
of a quantum system $\rho$ is described with a unitary operator $U$ by $U\rho
U^\dagger$, where $U^\dagger$ is the adjoint operator of $U$. The third
postulate states that the measurement is described by a set of measurement
operators $\{\pi_i\}_i$ with $\sum_i\pi_i^\dagger\pi_i=\mathsf I$, so that the
output of the measurement is $i$, with probability
$\tr(\pi_i^\dagger\pi_i\rho)$, leaving the sate of the system as
$\frac{\pi_i\rho\pi_i^\dagger}{\tr(\pi_i^\dagger\pi_i\rho)}$. The fourth
postulate states that from two systems $\rho$ and $\rho'$, the composed one can
be described by the tensor product of those $\rho\otimes\rho'$.

Naturally, if we want to use the output of a measurement as a condition in the
classical control, we need to know that output. However, density matrices can
still be used as a way to compare processes before running them. For example the
process of tossing a coin, and according to its result, applying Z or not to a
balanced superposition, and the process of tossing a coin and not looking at its
result, may look quite different in most quantum programming languages. Yet both
processes output the same density matrix, and so they are indistinguishable.

In~\cite{SelingerMSCS04}, Selinger introduced a language of quantum flow charts,
and an interpretation of his language into a CPO of density matrices. After this
paper, the language of density matrices has been widely used in quantum
programming, e.g.
\cite{DHondtPanangadenMSCS06,FengYuYingJCSS13,FengDuanYingPOPL11,YingTOPLAS11,YingYingWuPOPL17}.
Indeed, the book ``Foundations of Quantum Programming'' \cite{FOQ} is entirely
written in the language of density matrices. Yet, as far as we know, no lambda
calculus for density matrix have been proposed.

Apart from the distinction of languages by how they treat the quantum states
(vectors in a Hilbert space or density matrices), we also can distinguish the
languages on how the control is considered: either quantumly or classically. The
idea of quantum data / classical control stated by Selinger in
\cite{SelingerMSCS04} induced a quantum lambda calculus in this paradigm
\cite{SelingerValironMSCS06}. Later, this calculus was the base to construct the
programming language Quipper
\cite{GreenLeFanulumsdaineRossSelingerValironPLDI13}, an embedded, scalable,
functional programming language for quantum computing. The concept of quantum
data / classical control declares that quantum computers will run in a
specialized device attached to a classical computer, and it is the classical
computer which will instruct the quantum computer what operations to perform
over which qubits, and then read the classical result after a measurement. It is
a direct consequence from the observation that quantum circuits are classical
(i.e. one cannot superpose circuits or measure them). Several studies have been
done under this paradigm, e.g. \cite{ AltenkirchGrattageLICS05,
  GreenLeFanulumsdaineRossSelingerValironPLDI13, PaganiSelingerValironPOPL14,
  SelingerValironMSCS06, ZorziMSCS16 }.

Dually to the quantum data / classical control paradigm, there is what we can
call the quantum data and control paradigm. The idea is to provide a
computational definition of the notions of vector space and bilinear functions.
In the realm of quantum walks, quantum control is not uncommon (e.g.
\cite{AmbainisBachNayakVishwanathWatrousSTOC01,AharonovAmbainisKempeVaziraniSTOC01}).
Also, several high-level languages on quantum control have been proposed in the
past (e.g.
\cite{AltenkirchGrattageLICS05,YingYuFeng12,YingYuFeng14,BadescuPanangadenQPL15}),
however, up to now, no complete lambda-calculus with quantum control have been
proposed. We benefit, though, from the long line of works in this direction
\cite{ArrighiDowekLMCS17,ArrighiDiazcaroValironIC17,ArrighiDiazcaroLMCS12,AssafDiazcaroPerdrixTassonValironLMCS14,DiazcaroPetitWoLLIC12}.

In this paper, we propose a quantum extension to the lambda calculus,
$\lambdens$, in the quantum data / classical control paradigm, where the quantum
data is given by density matrices, as first suggested by Selinger's
interpretation of quantum flow charts~\cite{SelingerMSCS04}. Then, we propose a modification of such a
calculus, called $\qlambdens$, in which we generalise the density matrices to
the classical control: That is, after a measurement, we take all the possible
outcomes in a kind of generalised density matrix of arbitrary terms. The control
does not become quantum, since it is not possible to superpose programs in the
quantum sense. However, we consider the density matrix of the mixed state of
programs arising from a measurement. Therefore, this can be considered as a kind
of probabilistic control, or even another way, perhaps weaker, of quantum
control.

\paragraph{Outline of the Paper.}
In Section~\ref{sec:lambdens} we introduce the typed calculus $\lambdens$, which
manipulates density matrices, and we give two interpretations of the calculus.
One where the terms are interpreted into a generalisation of mixed states, and
another where the terms are interpreted into density matrices. Then we prove
some properties of those interpretations. In Section~\ref{sec:qlambdens} we
introduce a modification of $\lambdens$, called $\qlambdens$, where the output
of a measurement produce a sum with all the possible outputs. We then extend the
interpretation of $\lambdens$ to accommodate $\qlambdens$, and prove its basic
properties. In Section~\ref{sec:correctness} we prove the Subject Reduction
(Theorem~\ref{thm:SR}) and Progress (Theorem~\ref{thm:Progress}) properties for
both calculi. In Section~\ref{sec:Examples} we give two interesting examples, in
both calculi. Finally, in Section~\ref{sec:Conclusion}, we conclude and discuss
some future work.

\section{Classical-control calculus with probabilistic
  rewriting}\label{sec:lambdens}
\subsection{Definitions}
The grammar of terms, given in Table~\ref{tab:Grammar}, have been divided in
three categories.

\begin{enumerate}
\item Standard lambda calculus terms: Variables from a set $\mathsf{Vars}$,
  abstractions and applications.
\item The four postulates of quantum mechanics, with the measurement postulate
  restricted to measurements in the computational basis\footnote{A
    generalisation to any arbitrary measurement can be considered in a future,
    however, for the sake of simplicity in the classical control, we consider
    only measurements in the computational basis, which is a common practice in
    quantum lambda
    calculi~\cite{PaganiSelingerValironPOPL14,SelingerValironSTQC09,SelingerValironMSCS06,ZorziMSCS16,DallagoMasiniZorziQPL09,DiazcaroDowek17}.}:
  $\rho^n$ to represent the density matrix of a quantum system. $U^nt$ to
  describe its evolution. $\pi^nt$ to measure it. $t\otimes t$ to describe the
  density matrix of a composite system (that is, a non entangled system composed
  of two subsystems).
\item Two constructions for the classical control: a pair $(b^m,\rho^n)$, where
  $b^m$ is the output of a measurement in the computational basis and $\rho^n$
  is the resulting density matrix, and the conditional $\mathsf{letcase}$
  construction reading the output of the measurement.
\end{enumerate}
\begin{table}[t]
  \begin{align*}
    t & := x\mid \lambda x.t \mid tt & \text{(Standard lambda calculus)}\\
      & \hspace{2mm}\mid \rho^n\mid U^n t\mid \pi^n t\mid t\otimes t & \text{(Quantum postulates)}\\
      & \hspace{2mm}\mid  (b^m,\rho^n)\mid \letcase xr{t,\dots,t} & \text{(Classical control)}
  \end{align*}
  where:
  \begin{itemize}
  \item $n,m\in\mathbb N$, $m\leq n$.
  \item $\rho^n$ is a density matrix of $n$-qubits, that is, a positive
    $2^n\times 2^n$-matrix with trace $1$.
  \item $b^m\in\mathbb N$, $0\leq b^m<2^m$.
  \item $\{t,\dots,t\}$ contains $2^m$ terms.
  \item $U^n$ is a unitary operator of dimension $2^n\times 2^n$, that is, a
    $2^n\times 2^n$-matrix such that $(U^n)^\dagger=(U^n)^{-1}$.
  \item $\pi^n = \{\pi_0,\dots,\pi_{2^n-1}\}$, describes a quantum measurement
    in the computational basis, where each $\pi_i$ is a projector operator of
    dimension $2^n$ projecting to one vector of the canonical base.
  \end{itemize}
  \caption{Grammar of terms of $\lambdens$.}
  \label{tab:Grammar}
\end{table}

The rewrite system, given in Table~\ref{tab:TRS}, is described by the relation
$\topr[p]$, which is a probabilistic relation where $p$ is the probability of
occurrence. If $U^m$ is applied to $\rho^n$, with $m\leq n$, we write
$\ext{U^m}$ for $U^m\otimes I^{n-m}$. Similarly, we write $\ext{\pi^m}$ when we
apply this measurement operator to $\rho^n$ for $\{\pi_0\otimes
I^{n-m},\dots,\pi_{2^m-1}\otimes I^{n-m}\}$. If the unitary $U^m$ needs to be
applied, for example, to the last $m$ qubits of $\rho^n$ instead of the first
$m$, we will need to use the unitary transformation $I^{n-m}\otimes U^m$
instead. And if it is applied to the qubits $k$ to $k+m$, then, we can use
$\ext{I^{k-1}\otimes U^m}$.

This rewrite system assumes that after a measurement, the result is known.
However, since we are working with density matrices we could also provide an
alternative rewrite system where after a measurement, the system turns into a
mixed state. We left this possibility for Section~\ref{sec:qlambdens}.

The type system, including the grammar of types and the derivation rules, is
given in Table~\ref{tab:TS}. The type system is affine, so variables can be used
at most once, forbiding from cloning a density matrix.

\begin{table}[t]
  \begin{align*}
    (\lambda x.t)r &\topr t[r/x]\\
    U^m\rho^n &\topr {\rho'}^n &\text{ with }{\rho'}^n=\ext{U^m}\rho^n\ext{U^m}^\dagger\\
    \pi^m\rho^n &\topr[p_i] (i,\rho^n_i) & \text{ with }
                                           \left\{
                                           \begin{array}{l}
                                             p_i=\tr(\ext{\pi_i}^\dagger\ext{\pi_i}\rho^n)\\
                                             \rho^n_i = \tfrac{\ext{\pi_i}\rho^n\ext{\pi_i}^\dagger}{p_i}
                                           \end{array}
    \right.\\
    \rho\otimes\rho' &\topr \rho'' & \text{ with }\rho''=\rho\otimes\rho'\\
    \letcase x{(b^m,\rho^n)}{t_0,\dots,t_{2^m-1}} &\topr t_{b^m}[\rho^n/x]
  \end{align*}
  \smallskip
  \[
    \infer{\lambda x.t\topr[p]\lambda x.r}{t\topr[p] r}\qquad \infer{ts\topr[p]
      rs}{t\topr[p] r}\qquad \infer{st\topr[p] sr}{t\topr[p] r}\qquad
    \infer{U^nt\topr[p] U^nr}{t\topr[p] r}\qquad
  \]
  \[
    \infer{\pi^n t\topr[p] \pi^n r}{t\topr[p] r}\qquad \infer{t\otimes s\topr[p]
      r\otimes s}{t\topr[p] r}\qquad \infer{s\otimes t\topr[p] s\otimes
      r}{t\topr[p] r}
  \]
  \[
    \infer{\letcase xt{s_0,\dots,s_n}\topr[p]\letcase
      xr{s_0,\dots,s_n}}{t\topr[p] r}
  \]
  \caption{Rewrite system for $\lambdens$.}
  \label{tab:TRS}
\end{table}
\begin{table}[t]
  \[
    A:= n\mid (m,n)\mid A\multimap A
  \]
  where $m\leq n\in\mathbb N$.
  \[
    \infer[\mathsf{ax}]{\Gamma,x:A\vdash x:A}{} \qquad
    \infer[\multimap_i]{\Gamma\vdash\lambda x.t:A\multimap B}{\Gamma,x:A\vdash
      t:B} \qquad \infer[\multimap_e]{\Gamma,\Delta\vdash tr:B}{\Gamma\vdash
      t:A\multimap B & \Delta\vdash r:A}
  \]
  \[
    \infer[\mathsf{ax}_\rho]{\Gamma\vdash\rho^n:n}{} \qquad
    \infer[\mathsf{u}]{\Gamma\vdash U^mt:n}{\Gamma\vdash t:n} \qquad
    \infer[\mathsf{m}]{\Gamma\vdash\pi^mt:(m,n)}{\Gamma\vdash t:n} \qquad
    \infer[\otimes]{\Gamma,\Delta\vdash t\otimes r:n+m}{\Gamma\vdash t:n &
      \Delta\vdash r:m}
  \]
  \[
    \infer[\mathsf{ax}_{\mathsf{am}}]{\Gamma\vdash (b^m,\rho^n):(m,n)}{} \quad
    \infer[\mathsf{lc}] {\Gamma\vdash\letcase xr{t_0,\dots,t_{2^m-1}}:A}
    {x:n\vdash t_0:A & \dots & x:n\vdash t_{2^m-1}:A & \Gamma\vdash r:(m,n) }
  \]
  \caption{Type system for $\lambdens$.}
  \label{tab:TS}
\end{table}

\begin{example}\label{ex:teleportation}
  The teleportation algorithm, while it is better described by pure states, can
  be expressed in the following way:

  Let $\beta_{00} = \frac
  12(\ket{00}\bra{00}+\ket{00}\bra{11}+\ket{11}\bra{00}+\ket{11}\bra{11})$.
  Then, the following term expresses the teleportation algorithm.
  \[
    \lambda x. \letcase y {\pi^2(\had^1(\cnot^2 (x\otimes\beta_{00})))} { y,
      \Z_3 y, \X_3 y, \Z_3\X_3 y }
  \]
  where $\Z_3=I\otimes I\otimes\Z^1$ and $\X_3=I\otimes I\otimes\X^1$.

  The type derivation is as follows.

  \noindent\scalebox{.69}{
    \parbox{\textwidth}{
      \[
        \infer[\multimap_i]{\vdash\lambda x.\letcase
          y{\pi^2(\had^1(\cnot^2(x\otimes\beta_{00})))} {y, \Z_3 y, \X_3 y,
            \Z_3\X_3 y}:1\multimap 3} { \infer[\mathsf{lc}]{x:1\vdash\letcase
            y{\pi^2(\had^1(\cnot^2(x\otimes\beta_{00})))} {y, \Z_3 y, \X y,
              \Z_3\X_3 y}:3} { \infer[\mathsf{ax}]{y:3\vdash y:3}{} &
            \infer[\mathsf{u}]{y:3\vdash \Z_3 y:3}{
              \infer[\mathsf{ax}]{y:3\vdash y:3}{} } &
            \infer[\mathsf{u}]{y:3\vdash \X_3 y:3}{
              \infer[\mathsf{ax}]{y:3\vdash y:3}{} } &
            \infer[\mathsf{u}]{y:3\vdash \Z_3\X_3 y:3}{
              \infer[\mathsf{u}]{y:3\vdash \X_3 y:3}{
                \infer[\mathsf{ax}]{y:3\vdash y:3}{} } } &
            \infer[\mathsf{m}]{x:1\vdash\pi^2(\had^1(\cnot^2(x\otimes\beta_{00}))):(2,3)}
            { \infer[\mathsf{u}]{x:1\vdash\had^1(\cnot^2(x\otimes\beta_{00})):3}
              { \infer[\mathsf{u}]{x:1\vdash\cnot^2(x\otimes\beta_{00}):3} {
                  \infer[\otimes]{x:1\vdash x\otimes\beta_{00}:3} {
                    \infer[\mathsf{ax}]{x:1\vdash x:1}{} &
                    \infer[\mathsf{ax}_\rho]{\vdash\beta_{00}:2}{} } } } } } }
      \]
    } }
\end{example}

\subsection{Interpretation}\label{sec:DenSem}
We give two interpretations for terms. One, noted by $\fsem[]{\cdot}$, is the
interpretation of terms into density matrices and functions upon them, and the
other, noted by $\sem[]{\cdot}$, is a more fine-grained interpretation,
interpreting terms into a generalisation of mixed states. In particular, we want
$\sem[]{\pi^n\rho^n}\!=\!\{(\tr(\pi_i^\dagger\pi_i\rho^n),\frac{\pi_i\rho^n\pi_i^\dagger}{\tr(\pi_i^\dagger\pi_i\rho^n)})\}_i$,
while $\fsem[]{\pi^n\rho}=\sum_i\pi_i\rho^n\pi_i^\dagger$. However, since the
$\mathsf{letcase}$ construction needs also to distinguish each possible result
of a measurement, we will carry those results in the interpretation
$\sem[]{\cdot}$, making it a set of triplets instead of a set of tuples.

Let $\mathbb N^\varepsilon=\mathbb N_0\cup\{\varepsilon\}$, so terms are
interpreted into sets of triplets $(p,b,e)$ with $p\in\mathbb R_+^{\leq 1}$,
representing the probability, $b\in\mathbb N^\varepsilon$, representing the
output of a measurement if it occurred, and $e\in\tsem A$ for some type $A$ and
an interpretation $\tsem{\cdot}$ on types yet to define. In addition, we
consider that the sets $\{\dots,(p,b,e),(q,b,e),\dots\}$ and
$\{\dots,(p+q,b,e),\dots\}$ are equal. Finally, we define the weight function as
$\w{\{(p_i,b_i,e_i)\}_i}=\sum_ip_i$. We are interested in sets $S$ such that
$\w{S}=1$.

The interpretation of types is given in Table~\ref{tab:IntTy}. $\D$ is the set
of density matrices of $n$-qubits, that is $\D=\{\rho\mid\rho\in\mathcal
M^+_{2^n\times 2^n}\textrm{ such that } \tr(\rho)=1\}$, where $\mathcal
M^+_{2^n\times 2^n}$ is the set of positive matrices of size $2^n\times 2^n$.
$P(b,A)$ is the following property: $[A=(m,n)\implies
b\neq\varepsilon]$.
We also establish the convention that
$P(\{(p_i,b_i,e_i)\}_i,A)=\bigwedge_iP(b_i,A)$. Finally, we write
$\trd{S}=\{e\mid (p,b,e)\in S\}$.

\begin{table}[t]
  \begin{align*}
    \tsem{n} &= \D\\
    \tsem{(m,n)} &=\D\\
    \tsem{A\multimap B} &= \{f\mid \forall e\in\tsem A, \forall b\in\mathbb N^\varepsilon \textrm{ s.t. }P(b,A),\\
             &\qquad\trd{f(b,e)}\subseteq\tsem B, \w{f(b,e)}=1\textrm{ and }P(f(b,e),B)\}
  \end{align*}
  \caption{Interpretation of types}
  \label{tab:IntTy}
\end{table}

Let $E=\bigcup_{A\in\mathsf{Types}}\tsem A$. We denote by $\theta$ to a
valuation $\mathsf{Vars}\to\mathbb N^\varepsilon\times E$. Then, we define the
interpretation of terms with respect to a given valuation $\theta$ in
Table~\ref{tab:IntTerms}.

\begin{table}[t]
  \begin{align*}
    \sem x &= \{(1, b,e)\}\textrm{ where }\theta(x)=(b,e)\\
    \sem{\lambda x.t} &= \{(1,\varepsilon,\para\mapsto\sem[\theta,x=\para]t)\}\\
    \sem{tr} &= \begin{aligned}[t]
      \{(p_iq_jh_{ijk},b''_{ijk},g_{ijk})\mid
      & \sem r=\{(p_i,b_i,e_i)\}_i,\\
      &\sem t=\{(q_j,b'_j,f_j)\}_j \textrm{ and }\\
      &f_j(b_i,e_i)=\{(h_{ijk},b''_{ijk},g_{ijk})\}_k\}\\ 
    \end{aligned}\\
    \sem{\rho^n} &= \{(1,\emptybit,\rho^n)\}\\
    \sem{U^nt} &=\{(p_i,\emptybit,\ext{U^n}\rho_i\ext{U^n}^\dagger)\mid \sem t=\{(p_i,b_i,\rho_i)\}_i\}\\
    \sem{\pi^mt} &=\{(p_j\tr(\ext{\pi_i}^\dagger\ext{\pi_i}\rho_j),i,\frac{\ext{\pi_i}\rho_j\ext{\pi_i}^\dagger}{\tr(\ext{\pi_i}^\dagger\ext{\pi_i}\rho_j)})\mid \sem t=\{(p_j,b_j,\rho_j)\}_j\}\\
    \sem{t\otimes r} &=
                       \begin{aligned}[t]
                         \{(p_iq_j,\emptybit,\rho_i\otimes \rho'_j)\mid& \sem t=\{(p_i,b_i,\rho_i)\}_i\textrm{ and }\\
                         &\sem r=\{(q_j,b'_j,\rho'_j)\}_j\}
                       \end{aligned}
    \\
    \sem{(b^m,\rho^n)} &= \{(1,b^m,\rho^n)\}\\
    \omit\rlap{$\sem{\letcase xr{t_0,\dots,t_{2^m-1}}} =
    \begin{aligned}[t]
      & \{(p_iq_{ij},b'_{ij},e_{ij})\mid \\
      &\quad\sem r=\{(p_i,b_i,\rho_i)\}_i\textrm{ and }\\
      &\quad\sem[\theta,x=(\emptybit,\rho_i)]{t_{b_i}}=\{(q_{ij},b'_{ij},e_{ij})\}_j\}
    \end{aligned}$}
  \end{align*}
  \caption{Interpretation of terms}
  \label{tab:IntTerms}
\end{table}

\begin{definition}
  $\theta\vDash\Gamma$ if and only if, for all $x:A\in\Gamma$, $\theta(x)=(b,e)$
  with $e\in\tsem A$, and $P(b,A)$.
\end{definition}

Lemma~\ref{lem:listMN} states that a term with type $(m,n)$ will be the result of a measurement, and hence, its
interpretation will carry the results $b_i\neq\emptybit$.
\begin{lemma}\label{lem:listMN}
  Let $\Gamma\vdash t:(m,n)$, $\theta\vDash\Gamma$, and $\sem t$
  be well-defined. Then, $\sem t=\{(p_i,b_i,e_i)\}_i$ with $b_i\neq\emptybit$
  and $e_i\in\D$.
\end{lemma}
\begin{proof}
  By induction on the type derivation{ (cf.~Appendix~\ref{proof:listMN})}.\qed
\end{proof}

Lemma~\ref{lem:wellDefined} states that the interpretation of a typed term is
well-defined.
\begin{lemma}
  \label{lem:wellDefined}
  If $\Gamma\vdash t:A$ and $\theta\vDash\Gamma$, then $\w{\sem t}=1$, and
  $\trd{\sem t}\subseteq\tsem A$.
\end{lemma}
\begin{proof}
  By induction on $t${ (cf.~Appendix~\ref{proof:wellDefined})}.\qed
\end{proof}

Since the interpretation $\sem[]{\cdot}$ of a term is morally a mixed state, the
interpretation $\fsem[]{\cdot}$, which should be the density matrix of such a
state, is naturally defined using the interpretation $\sem[]{\cdot}$.
\begin{definition}
  Let $e\in\tsem A$ for some $A$, $\theta$ a valuation, and $t$ be a term such
  that $\sem t=\{(p_i,b_i,e_i)\}_i$.
  We state the convention that $\para\mapsto\sum_ip_ie_i=\sum_ip_i(\para\mapsto e_i)$.
  We define $\conv e$ and $\fsem t$ by mutual recursion as follows:
  \begin{align*}
    \conv{\rho} &=\rho\\
    \conv{\para\mapsto\sem[\theta,x=\para]t} &=\para\mapsto\fsem[\theta,x=\para]t\\
    \fsem t &=\sum_ip_i\conv{e_i}
  \end{align*}
\end{definition}

\begin{lemma}[Substitution]\label{lem:substitutionSem}
  Let $\sem r=\{(p_i,b_i,e_i)\}_i$, then
  \[
    \fsem{t[r/x]}=\sum_ip_i\fsem[\theta,x=(b_i,e_i)]t
  \]
\end{lemma}
\begin{proof}
  By induction on $t$. However, we enforce the hypothesis by also showing that
  if $\sem[\theta,x=(b_i,e_i)] t\!=\!\{(q_{ij},b'_{ij},\rho_{ij})\}_j$, then
  $\sem{t[r/x]}\!=\!\{(p_iq_{ij},b'_{ij},\rho_{ij})\}_{ij}$. We use five
  auxiliary results (cf.~{Appendix~\ref{proof:substitutionSem}}).\qed
\end{proof}

Theorem~\ref{thm:IntRed} shows how the interpretation $\fsem[]{\cdot}$ of a term
relates to all its reducts.

\begin{theorem}\label{thm:IntRed}
  If $\Gamma\vdash t:A$, $\theta\vDash\Gamma$ and $t\topr[p_i] r_i$, with
  $\sum_ip_i=1$, then $\fsem t=\sum_i p_i \fsem{r_i}$.
\end{theorem}
\begin{proof}
  By induction on the relation $\topr[p]${ (cf.~Appendix~\ref{proof:IntRed}).}\qed
\end{proof}

\section{Probabilistic-control calculus with no-probabilistic
  rewriting}\label{sec:qlambdens}
\subsection{Definitions}
In the previous sections we have presented an extension to lambda calculus to
handle density matrices. The calculus could have been done using just vectors,
because the output of a measurement is not given by the density matrix of the
produced mixed state, instead each possible output is given with its
probability. In this section, we give an alternative presentation, named
$\qlambdens$, where we can make the most of the density matrices setting.

In Table~\ref{tab:ExtGrammar} we give a modified grammar of terms for
$\qlambdens$ in order to allow for linear combination of terms. We follow the
grammar of the algebraic
lambda-calculi~\cite{ArrighiDowekLMCS17,AssafDiazcaroPerdrixTassonValironLMCS14,VauxMSCS09}.

\begin{table}[t]
  \begin{align*}
    t & := x\mid \lambda x.t \mid tt & \text{(Standard lambda calculus)}\\
      & \hspace{2mm}\mid \rho^n\mid U^n t\mid \pi^n t\mid t\otimes t & \text{(Quantum postulates)}\\
      & \hspace{2mm}\mid \sum_{i=1}^n p_i t_i\mid  \qletcase xr{t,\dots,t} & \text{(Probabilistic control)}
  \end{align*}
  where $p_i\in (0,1]$, $\sum_{i=1}^n p_i = 1$, and $\sum$ is considered modulo
  associativity and commutativity (cf.~for example~\cite{ArrighiDowekLMCS17}).
  \caption{Grammar of terms of $\qlambdens$.}
  \label{tab:ExtGrammar}
\end{table}

The new rewrite system is given by the non-probabilistic relation $\toD$,
described in Table~\ref{tab:ExtTRS}. The measurement does not reduce, unless it
is the parameter of a $\mathsf{letcase}^\circ$. Therefore, if only a measurement
is needed, we can encode it as:
\[
  \qletcase{x}{\pi^m\rho^n}{x,\dots,x}\toD \sum_ip_i\rho_i^n\toD \rho'
\]
where $\rho'=\sum_i\ext{\pi_i}\rho^n\ext{\pi_i}^\dagger$. The rationale is that
in this version of the calculus, we can never look at the result of a
measurement. It will always produce the density matrix of a mixed-state. As a
consequence, the $\mathsf{letcase}^\circ$ constructor rewrites to a sum of
terms.

\begin{table}[t]
  \begin{align*}
    (\lambda x.t)r &\toD t[r/x]\\
    \qletcase x{\pi^m\rho^n}{t_0,\dots,t_{2^m}-1} &\toD\sum\limits_i p_i t_i[\rho^n_i/x]
                                                  & \textrm{with }\left\{ \begin{array}{l}
                                                                            \rho_i^n=\frac{\ext{\pi_i}\rho^n\ext{\pi_i}^\dagger}{p_i}\\
                                                                            p_i=\tr(\ext{\pi_i}^\dagger\ext{\pi_i}\rho^n)
                                                                          \end{array}\right.\\
    U^m\rho^n &\toD {\rho'}^n &\textrm{with }\ext{U^m}\rho^n\ext{U^m}^\dagger={\rho'}^n\\
    \rho\otimes\rho'&\toD \rho'' & \text{with }\rho''=\rho\otimes\rho'\\
    \sum_i p_i\rho_i&\toD \rho' & \text{with }\rho'=\sum_ip_i\rho_i\\
    \sum_i p_i t &\toD t\\
    (\sum_i p_i t_i)r &\toD \sum_i p_i (t_ir)
  \end{align*}
  \[
    \infer{\lambda x.t\toD\lambda x.r}{t\toD r}\qquad \infer{ts\toD rs}{t\toD
      r}\qquad \infer{st\toD sr}{t\toD r}\qquad \infer{U^nt\toD U^nr}{t\toD
      r}\qquad
  \]
  \[
    \infer{\pi^n t\toD \pi^n r}{t\toD r}\qquad \infer{t\otimes s\toD r\otimes
      s}{t\toD r}\qquad \infer{s\otimes t\toD s\otimes r}{t\toD r}
  \]
  \[
    \infer[\scriptstyle(\forall i\neq j, t_i=r_j)]{\sum_{i=1}^n p_it_i\toD
      \sum_{i=1}^n p_ir_i} {t_j\toD r_j}
  \]
  \[
    \infer {\qletcase xt{s_0,\dots,s_{2^m-1}}\toD\qletcase
      xr{s_0,\dots,s_{2^m-1}}} {t\toD r}
  \]
  \caption{Rewrite system of $\qlambdens$.}
  \label{tab:ExtTRS}
\end{table}

The type system for $\qlambdens$, including the grammar of types and the
derivation rules, is given in Table~\ref{tab:ExtTS}. The only difference with
the type system of $\lambdens$ (cf.~Table~\ref{tab:TS}), is that rule
$\mathsf{ax}_{\mathsf{am}}$ is no longer needed, since $(b^m,\rho^n)$ is not in
the grammar of $\qlambdens$, and there is a new rule ($+$) typing the
generalised mixed states. We use the symbol $\Vdash$ for $\qlambdens$ to
distinguish it from $\vdash$ used in $\lambdens$.

\begin{table}[t]
  \[
    A:= n\mid (m,n)\mid A\multimap A
  \]
  where $m\leq n\in\mathbb N$.
  \[
    \infer[\mathsf{ax}]{\Gamma,x:A\Vdash x:A}{} \qquad
    \infer[\multimap_i]{\Gamma\Vdash\lambda x.t:A\multimap B}{\Gamma,x:A\Vdash
      t:B} \qquad \infer[\multimap_e]{\Gamma,\Delta\Vdash tr:B}{\Gamma\Vdash
      t:A\multimap B & \Delta\Vdash r:A}
  \]
  \[
    \infer[\mathsf{ax}_\rho]{\Gamma\Vdash\rho^n:n}{} \quad\
    \infer[\mathsf{u}]{\Gamma\Vdash U^mt:n}{\Gamma\Vdash t:n} \quad\
    \infer[\mathsf{m}]{\Gamma\Vdash\pi^mt:(m,n)}{\Gamma\Vdash t:n} \quad\
    \infer[\otimes]{\Gamma,\Delta\Vdash t\otimes r:n+m}{\Gamma\Vdash t:n &
      \Delta\Vdash r:m}
  \]
  \[
    \infer[\mathsf{lc}] {\Gamma\Vdash\qletcase xr{t_0,\dots,t_{2^m-1}}:A}
    {x:n\Vdash t_0:A & \dots & x:n\Vdash t_{2^m-1}:A & \Gamma\Vdash r:(m,n) }
  \]
  \[
    \infer[+]{\Gamma\Vdash\sum_{i=1}^n p_i t_i:A} { \Gamma\Vdash t_1:A &\dots&
      \Gamma\Vdash t_n:A & \sum_{i=1}^np_i=1 }
  \]
  \caption{Type system for $\qlambdens$.}
  \label{tab:ExtTS}
\end{table}

\begin{example}
  \label{ex:telepExt}
  The teleportation algorithm expressed in $\lambdens$ in
  Example~\ref{ex:teleportation}, is analogous for $\qlambdens$, only changing
  the term $\mathsf{letcase}$ by $\mathsf{letcase}^\circ$. Also, the type
  derivation is analogous. The difference is in the reduction. Let $\rho$ be the
  density matrix of a given quantum state (mixed or pure). Let
  \[
    \rho^3_0 =\rho\otimes\beta_{00}\ ,\qquad \rho^3_1 = (\cnot\otimes
    I)\rho^3_0\ ,\qquad\text{and}\qquad \rho^3_2 = (\had\otimes I\otimes
    I)\rho^3_1
  \]
  The trace of the teleportation of $\rho$ in $\lambdens$ is the following:
  \begin{align}\nonumber
    &(\lambda x.\letcase y{\pi^2(\had^1(\cnot^2(x\otimes\beta_{00})))}{y,\Z_3y,\X_3y,\Z_3\X_3y})\rho\\\nonumber
    &\topr \letcase y{\pi^2(\had^1(\cnot^2(\rho\otimes\beta_{00})))}{y,\Z_3y,\X_3y,\Z_3\X_3y}\\\nonumber
    &\topr \letcase y{\pi^2(\had^1(\cnot^2\rho^3_0))}{y,\Z_3y,\X_3y,\Z_3\X_3y}\\\nonumber
    &\topr \letcase y{\pi^2(\had^1\rho^3_1)}{y,\Z_3y,\X_3y,\Z_3\X_3y}\\
    &\topr \letcase y{\pi^2\rho^3_2}{y,\Z_3y,\X_3y,\Z_3\X_3y}\label{eq:last}
  \end{align}
  From \eqref{eq:last}, there are four possible reductions. For $i=0,1,2,3$, let
  $p_i=\tr(\ext{\pi_i}^\dagger\ext{\pi_i}\rho^3_2)$ and
  $\rho^3_{3i}=\frac{\ext{\pi_i}\rho^3_2\ext{\pi_i}^\dagger}{p_i}$. Then,
  \begin{itemize}
  \item $\eqref{eq:last}\topr[p_0] \letcase
    y{(0,\rho^3_{30})}{y,\Z_3y,\X_3y,\Z_3\X_3y} \topr \rho^3_{30}=\rho$.
  \item $\eqref{eq:last}\topr[p_1] \letcase
    y{(1,\rho^3_{31})}{y,\Z_3y,\X_3y,\Z_3\X_3y} \topr \Z_3\rho^3_{31}\topr\rho$.
  \item $\eqref{eq:last}\topr[p_2] \letcase
    y{(2,\rho^3_{32})}{y,\Z_3y,\X_3y,\Z_3\X_3y} \topr \X_3\rho^3_{32}\topr\rho$.
  \item $\eqref{eq:last}\topr[p_3] \letcase
    y{(3,\rho^3_{33})}{y,\Z_3y,\X_3y,\Z_3\X_3y} \topr
    \Z_3\X_3\rho^3_{33}\topr\rho$.
  \end{itemize}

  On the other hand, the trace of the same term, in $\qlambdens$, would be
  analogous until \eqref{eq:last}, just using $\toD$ instead of $\topr$. Then:
  \[
    \eqref{eq:last}\toD p_0\rho +p_1\Z_3\rho^3_{31} +p_2\X_3\rho^3_{32}
    +p_3\Z_3\X_3\rho^3_{33} \toD^* \sum_{i=0}^3 p_i\rho^3_{30}
    \toD(\sum_{i=0}^3p_i)\rho\toD\rho
  \]
\end{example}

\subsection{Interpretation}\label{sec:ExtDenSem}
The interpretation of $\lambdens$ given in Section~\ref{sec:DenSem} considers
already all the traces. Hence, the interpretation of $\qlambdens$ can be
obtained from a small modification of it. We only need to drop the
interpretation of the term that no longer exists, $(b^m,\rho^n)$, and add an
interpretation for the new term $\sum_ip_it_i$ as follows:
\[
  \sem{\sum_ip_it_i}= \{(p_iq_{ij},b_{ij},e_{ij})\mid \sem{t_i} =
  \{(q_{ij},b_{ij},e_{ij})\}_j\}
\]
The interpretation of $\mathsf{letcase}^\circ$ is the same as the interpretation
of $\mathsf{letcase}$.

Then, we can prove a theorem (Theorem~\ref{thm:ExtIntRed}) for $\qlambdens$
analogous to Theorem~\ref{thm:IntRed}.

We need the following auxiliary Lemmas.
\begin{lemma}\label{lem:ExtSumFsem}
  If $\Gamma\Vdash t:A$ and $\theta\vDash\Gamma$, then $\fsem{\sum_ip_it_i} =
  \sum_ip_i\fsem{t_i}$
\end{lemma}
\begin{proof}
  Let $\sem{t_i}=\{(q_{ij},b_{ij},e_{ij})\}_j$. Then, we have
  $\fsem{\sum_ip_it_i}=\sum_{ij}p_iq_{ij}e_{ij}=\sum_ip_i\sum_jq_{ij}e_{ij}=\sum_ip_i\fsem{t_i}$.
  \qed
\end{proof}

\begin{lemma}
  \label{lem:ExtSubstitutionSem}
  Let $\sem r=\{(p_i,b_i,e_i)\}_i$, then
  $\fsem{t[r/x]}=\sum_ip_i\fsem[\theta,x=(b_i,e_i)]t$.
\end{lemma}
\begin{proof}
  The proof of the analogous Lemma~\ref{lem:substitutionSem} in $\lambdens$
  follows by induction on $t$. Since the definition of $\sem[]{\cdot}$ is the
  same for $\lambdens$ than for $\qlambdens$, we only need to check the only
  term of $\qlambdens$ which is not a term of $\lambdens$: $\sum_jq_jt_j$. Using
  Lemma~\ref{lem:ExtSumFsem}, and the induction hypothesis, we have
  $\fsem{(\sum_jq_jt_j)[r/x]}=\fsem{\sum_jq_j(t_j[r/x])}
  =\sum_jq_j\fsem{t_j[r/x]} =\sum_jq_j\sum_ip_i\fsem[\theta,x=(b_i,e_i)]{t_j}
  =\sum_ip_i\fsem[\theta,x=(b_i,e_i)]{\sum_jq_jt_j}$. \qed
\end{proof}

\begin{theorem}\label{thm:ExtIntRed}
  If $\Gamma\Vdash t:A$, $\theta\vDash\Gamma$ and $t\toD r$, then $\fsem t=\fsem
  r$.
\end{theorem}
\begin{proof}
  By induction on the relation $\toD$. Rules $(\lambda x.t)r\toD t[r/x]$,
  $U^m\rho^n\toD \rho'$ and $\rho\otimes\rho'\toD\rho''$ are also valid rules
  for relation $\topr[1]$, and hence the proof of these cases are the same than
  in Theorem~\ref{thm:IntRed}. 
  \qed
\end{proof}

\section{Subject Reduction and Progress}\label{sec:correctness}
In this section we state and prove the subject reduction and progress properties
on both, $\lambdens$ and $\qlambdens$ (Theorems~\ref{thm:SR}
and~\ref{thm:Progress} respectively).

\needspace{2em}
\begin{lemma}
  [Weakening]\label{lem:weak}~
  \begin{itemize}
  \item If $\Gamma\vdash t:A$ and $x\notin FV(t)$, then $\Gamma,x:B\vdash t:A$.
  \item If $\Gamma\Vdash t:A$ and $x\notin FV(t)$, then $\Gamma,x:B\Vdash t:A$.
  \end{itemize}
\end{lemma}
\begin{proof}
  By a straightforward induction on the derivation of $\Gamma\vdash t:A$ and on
  $\Gamma\Vdash t:A$.\qed
\end{proof}

\begin{lemma}[Strengthening]
  \label{lem:streng}~
  \begin{itemize}
  \item If $\Gamma,x:A\vdash t:B$ and $x\notin FV(t)$, then $\Gamma\vdash t:B$.
  \item If $\Gamma,x:A\Vdash t:B$ and $x\notin FV(t)$, then $\Gamma\Vdash t:B$.
  \end{itemize}
\end{lemma}
\begin{proof}
  By a straightforward induction on the derivation of $\Gamma,x:A\vdash t:B$ and
  $\Gamma,x:A\Vdash t:B$.\qed
\end{proof}

\begin{lemma}[Substitution]\label{lem:substitution}~
  \begin{itemize}
  \item If $\Gamma,x:A\vdash t:B$ and $\Delta\vdash r:A$ then
    $\Gamma,\Delta\vdash t[r/x]:B$.
  \item If $\Gamma,x:A\Vdash t:B$ and $\Delta\Vdash r:A$ then
    $\Gamma,\Delta\Vdash t[r/x]:B$.
  \end{itemize}
\end{lemma}
\begin{proof}
  By induction on $t${ (cf.~Appendix~\ref{proof:substitution})}.
  \qed
\end{proof}

\begin{theorem}
  [Subject reduction]
  \label{thm:SR}~
  \begin{itemize}
  \item If $\Gamma\vdash t:A$, and $t\topr[p] r$, then $\Gamma\vdash r:A$.
  \item If $\Gamma\Vdash t:A$, and $t\toD r$, then $\Gamma\Vdash r:A$.
  \end{itemize}
\end{theorem}
\begin{proof}
  By induction on the relations $\topr[p]$ and $\toD${ (cf.
    Appendix~\ref{proof:SR})}. \qed
\end{proof}

\begin{definition}[Values]~
  \begin{itemize}
  \item A value in $\lambdens$ is a term $v$ defined by the following grammar:
    \begin{align*}
      w &:= x\mid \lambda x.v \mid w\otimes w\\
      v &:= w\mid\rho^n\mid (b^m,\rho^n)
    \end{align*}
  \item A value in $\qlambdens$ (or \qvalue) is a term $v$ defined by the
    following grammar:
    \begin{align*}
      w &:= x\mid\lambda x.v \mid w\otimes w\mid \sum_ip_iw_i\textrm{ with }w_i\neq w_j\textrm{ if }i\neq j\\
      v &:= w\mid \rho^n
    \end{align*}
  \end{itemize}
\end{definition}

\begin{lemma}
  \label{lem:values}~
  \begin{enumerate}
  \item If $v$ is a value, then there is no $t$ such that $v\topr[p] t$ for any
    $p$.
  \item If $v$ is a \qvalue, then there is no $t$ such that $v\toD t$.
  \end{enumerate}
\end{lemma}
\begin{proof}
  By induction on $v$ in both cases{ (cf.~Appendix~\ref{proof:values})}.\qed
\end{proof}

\needspace{2em}
\begin{theorem}
  [Progress]\label{thm:Progress}~
  \begin{enumerate}
  \item If $\vdash t:A$, then either $t$ is a value or there exist $n$,
    $p_1,\dots,p_n$, and $r_1,\dots,r_n$ such that $t\topr[p_i]r_i$.
  \item If $\Vdash t:A$ and $A\neq(m,n)$, then either $t$ is a \qvalue\ or there
    exists $r$ such that $t\toD r$.
  \end{enumerate}
\end{theorem}
\begin{proof}
  We relax the hypotheses and prove the theorem for open terms as well. That is:
  \begin{enumerate}
  \item If $\Gamma\vdash t:A$, then either $t$ is a value, there exist $n$,
    $p_1,\dots,p_n$, and $r_1,\dots,r_n$ such that $t\topr[p_i]r_i$, or $t$
    contains a free variable, and $t$ does not rewrite.
  \item If $\Gamma\Vdash t:A$, then either $t$ is a \qvalue, there exists $r$
    such that $t\toD r$, or $t$ contains a free variable, and $t$ does not
    rewrite.
  \end{enumerate}
  In both cases, we proceed by induction on the type derivation{
    (cf.~Appendix~\ref{proof:Progress})}.\qed
\end{proof}

\section{Examples}\label{sec:Examples}
\begin{example}\label{ex:example}
  Consider the following experiment: Measure some $\rho$ and then toss a coin to
  decide whether to return the result of the measurement, or to give the result
  of tossing a new coin.
  \subsubsection*{The experiment in $\lambdens$.}
  This experiment can be implemented in $\lambdens$ as follows:
  \[
    \begin{split}
      (\letcase y{\pi^1\ket +\bra +}{\lambda x.x,\lambda x.\letcase w{\pi^1\ket +&\bra +}{w,w}})\\
      &(\letcase z{\pi^1\rho}{z,z})
    \end{split}
  \]
  \paragraph*{Trace:}
  We give one possible probabilistic trace. Notice that, by using different
  strategies, we would get different derivation trees. We will not prove
  confluence in this setting (cf.~\cite{DiazcaroMartinezLSFA17} for a full
  discussion on the notion of confluence of probabilistic rewrite systems), but
  we conjecture that such a property is meet.

  We use the following notations:
  \begin{align*}
    s &= \pi^1\ket +\bra +\\
    t_0 &= \lambda x.x\\
    t_1 &=\lambda x.\letcase ws{w,w}\\
    \rho &=\frac 34\ket 0\bra 0+\frac{\sqrt 3}4\ket 0\bra 1+\frac{\sqrt 3}4\ket 1\bra 0+\frac 14\ket 1\bra 1\\
    r_1 &=\letcase ys{t_0,t_1}\\
    r_2 &=\letcase z{\pi^1\rho}{z,z}\\
    \mathsf l_x &= \letcase y{(x,\ket x\bra x)}{y,y}\textrm{ with }x=0,1\\
    r_1^x &= \letcase y{(x,\ket x\bra x)}{t_0,t_1}\textrm{ with }x=0,1
  \end{align*}

  Using this notation, the probabilistic trace is given by the tree in
  Table~\ref{tab:Trace}. Therefore, with probability $\frac 58$ we get $\ket
  0\bra 0$, and with probability $\frac 38$ we get $\ket 1\bra 1$. Thus, the
  density matrix of this mixed state is $\frac 58\ket 0\bra 0+\frac 38\ket 1\bra
  1$.
  \begin{table}[t]
    \begin{center}
      \scalebox{.974}{ \parbox{1.01\textwidth}{\centering
          \begin{tikzpicture}
            \node (A) at (0,0) {$r_1r_2$}; \node (Ao) at (-2.75,-1) {$r_1\mathsf
              l_0$}; \node (B) at (-2.75,-2) {$r_1\ket 0\bra 0$}; \node (Bo) at
            (-4.25,-3) {$r_1^0\ket 0\bra 0$}; \node (Bt) at (-1.25,-3)
            {$r_1^1\ket 0\bra 0$}; \node (D) at (-4.25,-4) {$t_0\ket 0\bra 0$};
            \node (Db) at (-4.25,-5) {$\ket 0\bra 0$}; \node (E) at (-1.25,-4)
            {$t_1\ket 0\bra 0$}; \node (Hb) at (-2.25,-6) {$\mathsf l_0$}; \node
            (H) at (-2.25,-7) {$\ket 0\bra 0$}; \node (Ib) at (-.25,-6)
            {$\mathsf l_1$}; \node (I) at (-.25,-7) {$\ket 1\bra 1$}; \node (J)
            at (-1.25,-5) {$\letcase ys{y,y}$};
            \node (At) at (2.75,-1) {$r_1\mathsf l_1$}; \node (C) at (2.75,-2)
            {$r_1\ket 1\bra 1$}; \node (Co) at (1.25,-3) {$r_1^0\ket 1\bra 1$};
            \node (F) at (1.25,-4) {$t_0\ket 1\bra 1$}; \node (Fb) at (1.25,-5)
            {$\ket 1\bra 1$}; \node (Ct) at (4.25,-3) {$r_1^1\ket 1\bra 1$};
            \node (G) at (4.25,-4) {$\letcase ys{y,y}$}; \node (Kb) at (3.25,-6)
            {$\mathsf l_0$}; \node (K) at (3.25,-7) {$\ket 0\bra 0$}; \node (Lb)
            at (5.25,-6) {$\mathsf l_1$}; \node (L) at (5.25,-7) {$\ket 1\bra
              1$}; \node (M) at (4.25,-5) {$t_1\ket 1\bra 1$}; \draw[thick,->]
            (A) -- (Ao) node[midway,sloped,above] {$\frac 34$}; \draw[thick,->]
            (Ao) -- (B) node[midway,right] {$\scriptstyle 1$}; \draw[thick,->]
            (B) -- (Bo) node[midway,sloped,rotate=270,left] {$\frac 12$};
            \draw[thick,->] (Bo) -- (D) node[midway,right] {$\scriptstyle 1$};
            \draw[thick,->] (D) -- (Db) node[midway,right] {$\scriptstyle 1$};
            \draw[thick,->] (J) -- (Hb) node[midway,sloped,rotate=270,left]
            {$\frac 12$}; \draw[thick,->] (Hb) -- (H) node[midway,right]
            {$\scriptstyle 1$}; \draw[thick,->] (B) -- (Bt)
            node[midway,sloped,rotate=90,right] {$\frac 12$}; \draw[thick,->]
            (Bt) -- (E) node[midway,right] {$\scriptstyle 1$}; \draw[thick,->]
            (E) -- (J) node[midway,right] {$\scriptstyle 1$}; \draw[thick,->]
            (J) -- (Ib) node[midway,sloped,rotate=90,right] {$\frac 12$};
            \draw[thick,->] (Ib) -- (I) node[midway,right] {$\scriptstyle 1$};
            \draw[thick,->] (A) -- (At) node[midway,sloped,above] {$\frac 14$};
            \draw[thick,->] (At) -- (C) node[midway,right] {$\scriptstyle 1$};
            \draw[thick,->] (C) -- (Ct) node[midway,sloped,rotate=90,right]
            {$\frac 12$}; \draw[thick,->] (Ct) -- (G) node[midway,right]
            {$\scriptstyle 1$}; \draw[thick,->] (G) -- (M) node[midway,right]
            {$\scriptstyle 1$}; \draw[thick,->] (C) -- (Co)
            node[midway,sloped,rotate=270,left] {$\frac 12$}; \draw[thick,->]
            (Co) -- (F) node[midway,right] {$\scriptstyle 1$}; \draw[thick,->]
            (F) -- (Fb) node[midway,right] {$\scriptstyle 1$}; \draw[thick,->]
            (M) -- (Kb) node[midway,sloped,rotate=270,left] {$\frac 12$};
            \draw[thick,->] (Kb) -- (K) node[midway,right] {$\scriptstyle 1$};
            \draw[thick,->] (M) -- (Lb) node[midway,sloped,rotate=90,right]
            {$\frac 12$}; \draw[thick,->] (Lb) -- (L) node[midway,right]
            {$\scriptstyle 1$};
          \end{tikzpicture}
        } }
    \end{center}
    \caption{Trace of the $\lambdens$ term implementing the experiment of
      Example~\ref{ex:example}.}
    \label{tab:Trace}
  \end{table}

  \paragraph*{Typing:}
  \begin{equation}
    \label{eq:t0}
    \infer[\multimap_i]
    {y:1\vdash\lambda x.\letcase w{\pi^1\ket +\bra +}{w,w}:1\multimap 1}
    {
      \infer[\mathsf{lc}]
      {y:1,x:1\vdash\letcase w{\pi^1\ket +\bra +}{w,w}:1}
      {
        \infer[\mathsf{ax}]{y:1,x:1,w:1\vdash w:1}{}
        &
        \infer[\mathsf{ax}]{y:1,x:1,w:1\vdash w:1}{}
        &
        \infer[\mathsf{m}]
        {\vdash\pi^1\ket +\bra +:(1,1)}
        {
          \infer[\mathsf{ax}_\rho]{\vdash\ket +\bra +:1}{}
        }
      }
    }
  \end{equation}
  \begin{equation}
    \label{eq:letcase}
    \infer[\mathsf{lc}]
    {\vdash\letcase y{\pi^1\ket +\bra +}{t_0,t_1}:1\multimap 1}
    {
      \infer[\multimap_i]
      {y:1\vdash\lambda x.x:1\multimap 1}
      {\infer[\mathsf{ax}]{y:1,x:1\vdash x:1}{}}
      &
      \infer[\eqref{eq:t0}]
      {y:1\vdash t_1:1\multimap 1}
      {\vdots}
      &
      \infer[\mathsf{m}]
      {\vdash\pi^1\ket +\bra +:(1,1)}
      {\infer[\mathsf{ax}_\rho]{\vdash\ket +\bra +:1}{}}
    }
  \end{equation}

  \noindent\scalebox{.95}{
    \parbox{\textwidth}{
      \[
        \infer[\multimap_e] {\vdash(\letcase y{\pi^1\ket +\bra
            +}{t_0,t_1})(\letcase z{\pi^1\rho}{z,z}):1} {
          \infer[\eqref{eq:letcase}] {\vdash\letcase y{\pi^1\ket +\bra
              +}{t_0,t_1}:1\multimap 1} {\vdots} & \infer[\mathsf{lc}]
          {\vdash\letcase z{\pi^1\rho}{z,z}:1} { \infer[\mathsf{ax}]{z:1\vdash
              z:1}{} &
            \infer[\mathsf{m}]{\vdash\pi^1\rho:(1,1)}{\infer[\mathsf{ax}_\rho]{\vdash\rho:1}{}}
          } }
      \]
    } }

  \paragraph*{Interpretation:}
  \begin{align*}
    &\sem[\emptyset]s =\{(\frac 12,0,\ket 0\bra 0),(\frac 12,1,\ket 1\bra 1)\}\\
    &\sem[y=(\emptybit,\ket 0\bra 0)]{t_0} =\{(1,\emptybit,\para\mapsto\{(1,\mathsf b,\mathsf e)\})\}\\
    &\sem[y=(\emptybit,\ket 1\bra 1)]{t_1} =\{(1,\emptybit,\para\mapsto\{(\frac 12,\emptybit,\ket 0\bra 0),(\frac 12,\emptybit,\ket 1\bra 1)\})\}\\
    &\sem[\emptyset]{r_1} =\{(\frac 12,\emptybit,\para\mapsto \{(\frac 12,\emptybit,\ket 0\bra 0),(\frac 12,\emptybit,\ket 1\bra 1)\}),(\frac 12,\emptybit,\para\mapsto\{(1,\mathsf b,\mathsf e)\} )\}\\
    &\sem[\emptyset]{\pi^1\rho} =\{(\frac 34,0,\ket 0\bra 0),(\frac 14,1,\ket 1\bra 1)\}\\
    &\sem[\emptyset]{r_2} =\{(\frac 34,\emptybit,\ket 0\bra 0),(\frac 14,\emptybit,\ket 1\bra 1)\}\\
  \end{align*}
  Then,
  \begin{align*}
    \sem[\emptyset]{r_1r_2} =
    \{&(\frac 3{16},\emptybit,\ket 0\bra 0),(\frac 1{16},\emptybit,\ket 0\bra
        0),(\frac 3{16},\emptybit,\ket 1\bra 1),\\
      &(\frac 1{16},\emptybit,\ket 1\bra 1),(\frac 38,\emptybit,\ket 0\bra 0),(\frac 18,\emptybit,\ket 1\bra 1)\}
  \end{align*}
  Hence,
  \begin{align*}
    \fsem[\emptyset]{r_1r_2} &
                               = \frac 3{16}\ket 0\bra 0 +\frac 1{16}\ket 0\bra 0+\frac 3{16}\ket 1\bra 1 +\frac 1{16}\ket 1\bra 1 +\frac 38\ket 0\bra 0+\frac 18\ket1\bra 1\\
                             &=\frac 58\ket 0\bra 0+\frac 38\ket 1\bra 1
  \end{align*}

  \subsubsection*{The experiment in $\qlambdens$.} In $\qlambdens$, the example
  becomes:
  \[
    \begin{split}
      t:=(\qletcase y{\pi^1\ket +\bra +}{\lambda x.x,\lambda x.\qletcase w{&\pi^1\ket +\bra +}{w,w}})\\
      &(\qletcase z{\pi^1\rho}{z,z})
    \end{split}
  \]

  \paragraph*{Trace:} In this case the trace is not a tree, because the relation
  $\toD$ is not probabilistic. We use the same $\rho$ as before: $\frac 34\ket
  0\bra 0+\frac{\sqrt 3}4\ket 1\bra 0+\frac{\sqrt 3}4\ket 0\bra 1+\frac 14\ket
  1\bra 1$.

  \begin{align*}
    t& \toD (\qletcase y{\pi^1\ket +\bra +}{\lambda x.x,\lambda x.\qletcase w{\pi^1\ket +\bra +}{w,w}})\\
     &\qquad(\frac 34\ket 0\bra 0+\frac 14\ket 1\bra 1)\\
     &\toD
       ( \frac 12\lambda x.x+\frac 12 \lambda x.\qletcase w{\pi^1\ket +\bra +}{w,w})(\frac 34\ket 0\bra 0+\frac 14\ket 1\bra 1)\\
     &\toD(\frac 12\lambda x.x+\frac 12(\lambda x.\frac 12\ket 0\bra 0+\frac 12\ket 1\bra 1))(\frac 34\ket 0\bra 0+\frac 14\ket 1\bra 1)\\
     &\toD
       \frac 12((\lambda x.x)(\frac 34\ket 0\bra 0+\frac 14\ket 1\bra 1))\\
     &\qquad+
       \frac 12((\lambda x.\frac 12\ket 0\bra 0+\frac 12\ket 1\bra 1)(\frac 34\ket 0\bra 0+\frac 14\ket 1\bra 1))\\
     &\toD
       \frac 12((\lambda x.x)(\frac 34\ket 0\bra 0+\frac 14\ket 1\bra 1))
       +\frac 12(\frac 12\ket 0\bra 0+\frac 12\ket 1\bra 1)\\
     &\toD
       \frac 12(\frac 34\ket 0\bra 0+\frac 14\ket 1\bra 1)
       +\frac 12(\frac 12\ket 0\bra 0+\frac 12\ket 1\bra 1)\\
     &\toD\frac 58\ket 0\bra 0+\frac 38\ket 1\bra 1
  \end{align*}

  \paragraph*{Typing and Interpretation:} Since $t$ does not contain sums, its
  typing is analogous to the term in $\lambdens$, as well as the interpretation.
\end{example}

\begin{example}\label{ex:page371}
  In \cite[p 371]{NielsenChuang00} there is an example of the freedom in the
  operator-sum representation by showing two quantum operators, which are
  actually the same. One is the process of tossing a coin and, according to its
  results, applying $I$ or $Z$ to a given qubit The second is the process
  performing a projective measurement with unknown outcome to the same qubit.
  These operations can be encoded in $\lambdens$ by:
  \begin{align*}
    O_1&=\lambda y.\letcase x{\pi^1\ket +\bra +}{y,\Z y}\\
    O_2&=\lambda y.\letcase x{\pi^1y}{x,x}
  \end{align*}
  with $\pi^1 = \{\ket 0\bra 0,\ket 1\bra 1\}$.

  Let us apply those operators to the qubit $\rho=\frac 34\ket 0\bra
  0+\frac{\sqrt 3}4\ket 0\bra 1+\frac{\sqrt 3}4\ket 1\bra 0+\frac 14\ket 1\bra
  1$. We can check that the terms $O_1\rho$ and $O_2\rho$ have different
  interpretations $\sem[]{\cdot}$. Let $\rho^-=\Z\rho\Z^\dagger$, then
  \begin{align*}
    \sem[\emptyset]{(\lambda y.\letcase x{\pi^1\ket +\bra +}{y,\Z y})\rho}
    &=\{(\frac 12,\emptybit,\rho),(\frac 12,\emptybit,\rho^-)\}\\
    \sem[\emptyset]{(\lambda y.\letcase x{\pi^1y}{x,x})\rho}
    &= \{(\frac 34,\emptybit,\ket 0\bra 0), (\frac 14,\emptybit,\ket 1\bra 1)\}    
  \end{align*}
  However, they have the same interpretation $\fsem[]{\cdot}$.
  \begin{align*}
    &\fsem[\emptyset]{(\lambda y.\letcase x{\pi^1\ket +\bra +}{y,\Z y})\rho}\\
    & = \frac 12\rho+\frac 12\rho^-\\
    & =\frac 34\ket 0\bra 0+\frac 14\ket 1\bra 1\\
    &=\fsem[\emptyset]{(\lambda y.\letcase x{\pi^1y}{x,x})\rho}
  \end{align*}
  The trace of $O_1\rho$ is given in Table~\ref{tab:TracesEx371-1}, and the
  trace of $O_2\rho$ in Table~\ref{tab:TracesEx371-2}.
  \begin{table}[t]
    \begin{center}
      \begin{tikzpicture}
        \node (A) at (0,0) {$O_1\rho$}; \node (B) at (0,-1.5) {$\letcase
          x{\pi^1\ket +\bra +}{\rho,\Z\rho}$}; \node (C) at (-2.5,-3) {$\letcase
          x{(0,\ket 0\bra 0)}{\rho,\Z\rho}$}; \node (D) at (2.5,-3) {$\letcase
          x{(1,\ket 1\bra 1)}{\rho,\Z\rho}$}; \node (E) at (-2.5,-4.5) {$\rho$};
        \node (F) at (2.5,-4.5) {$\Z\rho$}; \node (G) at (2.5,-6) {$\rho^-$};
        \draw[thick,->] (A) -- (B) node[midway,right] {$\scriptstyle 1$};
        \draw[thick,->] (B) -- (C) node[midway,sloped,above] {$\frac 12$};
        \draw[thick,->] (B) -- (D) node[midway,sloped,above] {$\frac 12$};
        \draw[thick,->] (C) -- (E) node[midway,right] {$\scriptstyle 1$};
        \draw[thick,->] (D) -- (F) node[midway,right] {$\scriptstyle 1$};
        \draw[thick,->] (F) -- (G) node[midway,right] {$\scriptstyle 1$};
      \end{tikzpicture}
    \end{center}
    \caption{Trace of the terms $O_1\rho$ from Example~\ref{ex:page371} in
      $\lambdens$.}
    \label{tab:TracesEx371-1}
  \end{table}
  \begin{table}[t]
    \begin{center}
      \begin{tikzpicture}
        \node (Hp) at (0,-6) {$O_2\rho$}; \node (H) at (0,-7) {$\letcase
          x{\pi^1\rho}{x,x}$}; \node (I) at (-2.5,-8.5) {$\letcase x{(0,\ket
            0\bra 0)}{x,x}$}; \node (J) at (2.5,-8.5) {$\letcase x{(1,\ket 1\bra
            1)}{x,x}$}; \node (K) at (-2.5,-10) {$\ket 0\bra 0$}; \node (L) at
        (2.5,-10) {$\ket 1\bra 1$}; \draw[thick,->] (Hp) -- (H)
        node[midway,right] {$\scriptstyle 1$}; \draw[thick,->] (H) -- (I)
        node[midway,sloped,above] {$\frac 34$}; \draw[thick,->] (H) -- (J)
        node[midway,sloped,above] {$\frac 14$}; \draw[thick,->] (I) -- (K)
        node[midway,right] {$\scriptstyle 1$}; \draw[thick,->] (J) -- (L)
        node[midway,right] {$\scriptstyle 1$};
      \end{tikzpicture}
    \end{center}
    \caption{Trace of the term $O_2\rho$ from Example~\ref{ex:page371} in
      $\lambdens$.}
    \label{tab:TracesEx371-2}
  \end{table}
  The first term produces $\rho$, with probability $\frac 12$, and $\rho^-$,
  with probability $\frac 12$, while the second term produces either $\ket 0\bra
  0$ with probability $\frac 34$ or $\ket 1\bra 1$, with probability $\frac 14$.

  However, if we encode the same terms in $\qlambdens$, we can get both programs
  to produce the same density matrix:
  \begin{align*}
    O^\circ_1&=\lambda y.\qletcase x{\pi^1\ket +\bra +}{y,\Z y}\\
    O^\circ_2&=\lambda y.\qletcase x{\pi^1 y}{x,x}
  \end{align*}
  The traces of $O^\circ_1\rho$ and $O^\circ_2\rho$ are as follow:
  \[
    \begin{array}{rl@{\ }|@{\ }rl}
      &O^\circ_1\rho&
      &O^\circ_2\rho\\
      &=(\lambda y.\qletcase x{\pi^1\ket +\bra +}{y,\Z y})\rho&
      &=(\lambda y.\qletcase x{\pi^1 y}{x,x})\rho\\
      &\toD\qletcase x{\pi^1\ket +\bra +}{\rho,\Z\rho}&
      &\toD\qletcase x{\pi^1\rho}{x,x})\\
      &\toD(\frac 12\rho)+(\frac 12\Z\rho)&
      &\toD(\frac 34\ket 0\bra 0)+(\frac 14\ket 1\bra 1)\\
      &\toD(\frac 12\rho)+(\frac 12\rho^-)&
      &\toD\frac 34\ket 0\bra 0+\frac 14\ket 1\bra 1\\
      &\toD\frac 34\ket 0\bra 0+\frac 14\ket 1\bra 1&
      &    
    \end{array}
  \]

\end{example}

\section{Conclusions}\label{sec:Conclusion}
In this paper we have presented the calculus $\lambdens$, which is a quantum
data / classical control extension to the lambda calculus where the data is
manipulated by density matrices. The main importance of this calculus is its
interpretation into density matrices, which can equate programs producing the
same density matrices. Then, we have given a second calculus, $\qlambdens$,
where the density matrices are generalised to accommodate arbitrary terms, and
so, programs producing the same density matrices, rewrite to such a matrix,
thus, coming closer to its interpretation. The control of $\qlambdens$ is not
classical nor quantum, however it can be seen as a weaker version of the quantum
control approach. It is indeed not classical control because a generalised
density matrix of terms is allowed ($\sum_ip_it_i$). It is not quantum control
because superposition of programs are not allowed (indeed, the previous sum is
not a quantum superposition since all the $p_i$ are positive and so no
interference can occur). However, it is quantum in the sense that programs in a
kind of generalised mixed-states are considered. We preferred to call it {\em
  probabilistic control}.

As depicted in Example~\ref{ex:page371}, the calculus $\qlambdens$ allows to
represent the same operator in different ways. Understanding when two operators
are equivalent is important from a physical point of view: it gives insights on
when two different physical processes produce the same dynamics. To the best of
our knowledge, it is the first lambda calculus for density matrices.

\subsection*{Future work and open questions}
As pointed out by B\v{a}descu and Panangaden~\cite{BadescuPanangadenQPL15}, one
of the biggest issues with quantum control is that it does not accommodate well
with traditional features from functional programming languages like recursion.
Ying~\cite{Ying14} went around this problem by introducing a recursion based on
second quantisation. Density matrices are DCPOs with respect to the L\"owner
order. Is the form of weakened quantum control suggested in this paper monotone?
Can it be extended with recursion? Could this lead to a concrete quantum
programming language, like
Quipper~\cite{GreenLeFanulumsdaineRossSelingerValironPLDI13}?

All these open questions are promising new lines of research that we are willing
to follow. In particular, we have four ongoing works trying to answer some of
these questions:
\begin{itemize}
\item The most well studied quantum lambda calculus is, without doubt,
  Selinger-Valiron's $\lambda_q$~\cite{SelingerValironMSCS06}. Hence, we are
  working on the mutual simulations between $\lambdens$ and $\lambda_q$, and
  between $\qlambdens$ and a generalisation of $\lambda_q$ into mixed states.
\item We are also working on a first prototype of an implementation of
  $\qlambdens$.
\item We are studying extensions to both $\lambdens$ and $\qlambdens$ with
  recursion and with polymorphism.
\item Finally, we are studying a more sophisticated denotational semantics for
  both calculi than the one given in this paper. We hope such a semantics to be
  adequate and fully abstract.
\end{itemize}

\paragraph{Acknowledgements} We want to thank the anonymous reviewer for some
important references and suggestions on future lines of work.

  \newpage
  \appendix
  \section{Detailed Proofs of Section~\ref{sec:DenSem}}
  \subsection{Proof of Lemma~\ref{lem:listMN}}\label{proof:listMN}
  By induction on the type derivation.
  \begin{itemize}
  \item Let $\Gamma,x:(m,n)\vdash x:(m,n)$ as a
    consequence of rule $\mathsf{ax}$. Since $\theta\vDash\Gamma,x:(m,n)$, we have $\theta(x)=(b,e)$ with $e\in\D$ and
    $b\neq\emptybit$. Hence, $\sem x=\{(1,b,e)\}$.
  \item Let $\Gamma,\Delta\vdash tr:(m,n)$ as a consequence of
    $\Gamma\vdash t:\vec A\multimap(m,n)$, $\Delta\vdash r:A$ and rule
    $\multimap_e$. Since $\theta\vDash\Gamma,\Delta$, we have
    $\theta\vDash\Gamma$. Since $\sem{tr}$ is well-defined, we have $\sem
    r=\{(p_i,b_i,e_i)\}_i$, $\sem t=\{(q_j,b'_j,f_j)\}_i$, and
    $f_j(b_i,e_i)=\{(h_{ijk},b_{ijk}'',g_{ijk})\}_k$, and hence
    $\sem{tr}=\{(p_iq_jh_{ijk},b''_{ijk},g_{ijk})\}_{ijk}$. By the induction
    hypothesis, $f_i\in\tsem{\vec A\multimap(m,n)}$, so, by
    definition, $b_{ijk}''\neq\emptybit$ and $g_{ijk}\in\D$.

  \item Let $\Gamma\vdash\pi^m t:(m,n)$ as a consequence of $\Gamma\vdash t:n$
    and rule $\mathsf{m}$. Then we have that $\sem{\pi^mt}$ is equal to
    $\{p_j\tr(\ext{\pi_i}^\dagger\ext{\pi_i}\rho_j),i,\frac{\ext{\pi_i}\rho_j\ext{\pi_i}^\dagger}{\tr(\ext{\pi_i}^\dagger\ext{\pi_i}\rho_j)}\mid
    \sem t=\{(p_j,\emptybit,\rho_j)\}_j\}$, with $i\neq\emptybit$. Notice that
    $\frac{\ext{\pi_i}\rho_j\ext{\pi_i}^\dagger}{\tr(\ext{\pi_i}^\dagger\ext{\pi_i}\rho_j)}\in\D$.
  \item Let $\Gamma\vdash(b^m,\rho^n):(m,n)$ as a consequence of rule
    $\mathsf{ax}_{\mathsf{am}}$. Then, $\sem{(b^m,\rho^n)}=\{(1,b^m,\rho^n)\}$.
    Notice that $\rho^n\in\D$.
  \item Let $\Gamma\vdash\letcase xr{t_0,\dots,t_{2^{m'}-1}}:(m,n)$ as a consequence of $x:n'\vdash t_k:(m,n)$,
    for $k=0,\dots,2^{m'}-1$, $\Gamma\vdash r:(m',n')$ and rule $\mathsf{lc}$.
    Since $\sem{\letcase xr{t_0,\dots,t_{2^{m'}-1}}}$ is well-defined, $\sem
    r\!=\{(p_i,b_i,\rho_i^n)\}_i$,
    $\sem[\theta,x=(\emptybit,\rho_i^n)]{t_i}=\{(q_{ij},b_{ij}',e_{ij})\}_j$ and
    $\sem{\letcase xr{t_0,\dots,t_{2^{m'}-1}}}
    \!=\!\{(p_iq_{ij},b_{ij}',e_{ij})\}_{ij}$. By the induction hypothesis,
    $b_{ij}' \neq\emptybit$ and $e_{ij}\in\D$.
    \qed
  \end{itemize}

\subsection{Proof of Lemma~\ref{lem:wellDefined}}\label{proof:wellDefined}
By induction on $t$.
\begin{itemize}
\item Let $t=x$. Then $x:A\in\Gamma$. $\sem x=\{(1,b,e)\}$, with
  $\theta(x)=(b,e)$. Since $\theta\vDash\Gamma$ we have $e\in\tsem A$.
\item Let $t=\lambda x.r$. Then $A=B\multimap C$, $\Gamma,x:B\vdash r:C$ and
  $\sem{\lambda x.r}=\{(1,\emptybit,\para\mapsto\sem[\theta,x=\para]r)\}$. Let
  $e\in\tsem B$ and $b\in\mathbb N^\varepsilon$ such that $P(b,B)$. Then, since
  $\theta\vDash\Gamma$, we have $\theta,x=(b,e)\vDash\Gamma,x:B$. Hence, by the
  induction hypothesis, $\trd{\sem[\theta,x=(b,e)]r}\subseteq\tsem C$.
  Therefore, $\para\mapsto\sem[\theta,x=\para]r\in\tsem{B\multimap C}$.
\item Let $t=rs$. Then $\Gamma=\Gamma_1,\Gamma_2$, with $\Gamma_1\vdash
  r:B\multimap A$ and $\Gamma_2\vdash s:B$. Since $\theta\vDash\Gamma$, we have
  $\theta\vDash\Gamma_1$ and $\theta\vDash\Gamma_2$, so by the induction
  hypothesis, $\sem s=\{(p_i,b_i,e_i)\}_i$, $\sum_ip_i=1$ and $e_i\in\tsem B$.
  Similarly, $\sem r=\{(q_j,b'_j,f_j)\}_j$, $\sum_jq_j=1$ and
  $f_j\in\tsem{B\multimap A}$. Let
  $f_j(b_i,e_i)=\{(h_{ijk},b_{ijk}'',g_{ijk})\}_k$. Then,
  $\sem{rs}=\{(p_iq_jh_{ijk},b_{ijk}'',g_{ijk})\}_{ijk}$. By
  Lemma~\ref{lem:listMN}, we have $P(b,B)$, so by definition of
  $\tsem{B\multimap A}$, we have $\{g_{ijk}\}=\trd{f_j(b_i,e_i)}\subseteq\tsem
  A$ and $\w{f_j(b_i,e_i)}=1$. Notice that
  $\sum_{ijk}p_iq_jh_{ijk}=\sum_ip_i\sum_jq_j\sum_kh_{ijk}=1$.

\item Let $t=\rho^n$. Then $A=n$ and $\sem{\rho^n}=\{(1,\emptybit,\rho^n)\}$.
  Notice that $\rho^n\in\D=\tsem A$.
\item Let $t= U^mr$. Then $A=n$ and $\Gamma\vdash r:n$. By the induction
  hypothesis, $\sem r=\{(p_i,b_i,\rho^n_i)\}_i$, with $\sum_ip_i=1$ and
  $\rho^n_i\in\D$. Then, $\ext{U^m}\rho^n\ext{U^m}^\dagger\in\D=\tsem A$, and
  $\sem{U^mr}=\{(p_i,\emptybit,\ext{U^m}\rho^n\ext{U^m}^\dagger)\}$.
\item Let $t=\pi^m r$. Then $A=(m,n)$ and $\Gamma\vdash r:n$. By the induction
  hypothesis, $\sem r=\{(p_i,b_i,\rho^n_i)\}_i$, with $\sum_ip_i=1$ and
  $\rho^n_i\in\D$. Let $\pi^m = \{\pi_j\}_j$, and
  $q_{ij}=\tr(\ext{\pi_j}^\dagger\ext{\pi_j}\rho^n_i)$. So,$\sum_j q_{ij}=1$.
  Let
  $\rho^n_{ij}=\frac{\ext{\pi_j}\rho^n_i\ext{\pi_j}^\dagger}{q_{ij}}\in\D=\tsem{(m,n)}$.
  Then, we have $\sem{\pi^m r}=\{(p_iq_{ij},j,\rho^n_{ij})\}_{ij}$. Notice that
  $\sum_{ij}p_iq_{ij}=\sum_ip_i\sum_jq_{ij}=1$.
\item Let $t=r\otimes s$. Then $A=n+m$ and $\Gamma=\Gamma_1,\Gamma_2$, with
  $\Gamma_1\vdash r:n$ and $\Gamma_2\vdash s:m$. Since $\theta\vDash\Gamma$, we
  have $\theta\vDash\Gamma_1$ and $\theta\vDash\Gamma_2$. Then, by the induction
  hypothesis, $\sem r=\{(p_i,b_i,\rho^n_i)\}_i$, with $\sum_i p_i=1$ and
  $\rho^n_i\in\D$. Similarly, $\sem s=\{(q_j,b'_j,\rho^m_j)\}_j$, with $\sum_j
  q_j=1$ and $\rho^m_j\in\D[m]$. Hence, $\sem{r\otimes
    s}=\{(p_iq_j,\emptybit,\rho^n_i\otimes\rho^m_j)\}_{ij}$. Notice that
  $\rho^n_i\otimes\rho^m_j\in\D[n+m]$, and
  $\sum_{ij}p_iq_j=\sum_ip_i\sum_jq_j=1$.
\item Let $t=(b^m,\rho^n)$. Then $A=(m,n)$ and
  $\sem{(b^m,\rho^n)}=\{(1,b^m,\rho^n)\}$, with $\rho^n\in\D=\tsem A$.
\item Let $t=\letcase xr{s_0,\dots,s_{2^m-1}}$. Then, $x:n\vdash s_k:A$, for
  $k=0,\dots,2^m-1$, and $\Gamma\vdash r:(m,n)$. Hence, by the induction
  hypothesis, $\sem r=\{(p_i,b_i,\rho_i^n)\}_i$ with $\rho_i^n\in\D$ and
  $\sum_ip_i=1$. By Lemma~\ref{lem:listMN}, $b_i\neq\emptybit$, so
  $\theta,x=(\emptybit,\rho_i^n)\vDash x:n$ and hence, by the induction
  hypothesis,
  $\sem[\theta,x=(\emptybit,\rho_i^n)]{s_{b_i}}=\{(q_{ij},b_{ij}',e_{ij})\}_j$,
  with $\sum_jq_{ij}=1$ and $e_{ij}\in\tsem A$. Hence, $\sem{\letcase
    xr{s_0,\dots,s_{2^m-1}}} =\{(p_iq_{ij},b_{ij}',e_{ij})\}_{ij}$. Notice that
  i $\sum_{ij}p_iq_{ij}=\sum_ip_i\sum_jq_{ij}=1$. \qed
\end{itemize}

\subsection{Proof of
  Lemma~\ref{lem:substitutionSem}}\label{proof:substitutionSem}
\noindent We need the next four lemmas and the corollary as auxiliary results.

\begin{lemma}\label{lem:InterpAbs}
  If $\Gamma\vdash t:A\multimap B$ and $\theta\vDash\Gamma$, then we have $\fsem
  t=\sum_ip_i(\para\mapsto\fsem[\theta,x=\para]{r_i})$. \qed
\end{lemma}

\begin{lemma}\label{lem:InterpApp}
  If $\fsem t=\sum_ip_i(\para\mapsto\fsem[\theta',x=\para]{r_i})$ and $\sem
  s=\{(q_j,b_j,e_j)\}_j$, then we have
  $\fsem{ts}=\sum_{ij}p_iq_j\fsem[\theta',x=(b_j,e_j)]{r_i}$ \qed
\end{lemma}

\begin{corollary}
  \label{cor:InterpApp}
  If $\fsem t\!=\sum_ip_i(\para\mapsto\fsem[\theta',x=\para]{r_i})$ and $\fsem
  s=\sum_jq_je_j$, then
  $\fsem{ts}\!=\sum_{ij}p_iq_j\fsem[\theta',x=(\emptybit,e_j)]{r_i}$ \qed
\end{corollary}

\begin{lemma}\label{lem:InterpUM}
  If $\Gamma\vdash t:n$ and $\theta\vDash\Gamma$, then
  $\fsem{U^mt}=\ext{U^m}\fsem t\ext{U^m}^\dagger$ and
  $\fsem{\pi^mt}=\sum_i\ext{\pi_i}\fsem t\ext{\pi_i}^\dagger$. \qed
\end{lemma}

\begin{lemma}
  \label{lem:InterpTensor}
  If $\Gamma\vdash t:n$, $\Delta\vdash r:m$ and $\theta\vDash\Gamma,\Delta$,
  then $\fsem{t\otimes r}=\fsem t\otimes\fsem r$. \qed
\end{lemma}

\subsubsection*{Proof of Lemma~\ref{lem:substitutionSem}.}
We proceed by induction on $t$, however, we enforce the hypothesis by also
showing that if $\sem[\theta,x=(b_i,e_i)] t=\{(q_{ij},b'_{ij},\rho_{ij})\}_j$,
then $\sem{t[r/x]}=\{(p_iq_{ij},b'_{ij},\rho_{ij})\}_{ij}$.
\begin{itemize}
\item Let $t=x$. $\fsem{x[r/x]}=\fsem
  r=\sum_ip_i[e_i]=\sum_ip_i\fsem[\theta,x=(b_i,e_i)]{x}$.

  Notice that $\sem[\theta,x=(b_i,e_i)] x=\{(1,b_i,e_i)\}$ and $\sem
  r=\{(p_i,b_i,e_i)\}_i$.
\item Let $t=y$. Let $\theta(y)=(b',f)$. Then, $\fsem{y[r/x]}=\fsem
  y=f=\fsem[\theta,x=(b_i,e_i)]y=\sum_ip_i\fsem[\theta,x=(b_i,e_i)]y$.

  Notice that $\sem[\theta,x=(b_i,e_i)]
  y=\{(1,b',f)\}=\{(p_i,b',f')\}_i=\sem{y[r/x]}$.
\item Let $t=\lambda y.s$. Then, using the induction hypothesis, we have
  $\fsem{(\lambda y.s)[r/x]} =\fsem{\lambda y.s[r/x]}
  =\!\para\!\mapsto\!\fsem[\theta,y=\para]{s[r/x]}
  =\!\para\!\mapsto\!\sum_ip_i\fsem[\theta,y=\para,x=(b_i,e_i)]s
  =\sum_ip_i(\para\mapsto\fsem[\theta,y=\para,x=(b_i,e_i)]s)
  =\sum_ip_i\fsem[\theta,x=(b_i,e_i)]{\lambda y.s}$. Notice that $\lambda y.s$
  cannot have type $n$ or $(m,n)$.
\item Let $t=s_1s_2$
  \begin{itemize}
  \item Let $x\in FV(s_1)$. By Lemma~\ref{lem:InterpAbs}, we have
    $\fsem[\theta,x=(b_i,e_i)]{s_1}\!=\sum_jq_{ij}(\para\mapsto\fsem[\theta,x=(b_i,e_i),y=\para]{s'_{ij}})$.
    Then, by the induction hypothesis, $\fsem{s_1[r/x]}
    =\sum_ip_i\fsem[\theta,x=(b_i,e_i)]{s_1} =\sum_{ij} p_iq_{ij}(\para
    \mapsto\fsem[\theta,x=(b_i,e_i),y=\para]{s'_{ij}})$. Let $\sem{s_2}$
    $=\{(h_k,b'_k,f_k)\}_k$. Hence, by Lemma~\ref{lem:InterpApp}, we have that
    $\fsem{s_1[r/x]s_2}$
    $=\sum_{ijk}p_iq_{ij}h_k\fsem[\theta,x=(b_i,e_i),y=(b'_k,f_k)]{s'_{ij}}$.
    Since $x\notin FV(s_2)$, we also have that
    $\sem[\theta,x=(b_i,e_i)]{s_2}=\sem{s_2}$, so, also by
    Lemma~\ref{lem:InterpApp}, we have
    $\fsem[\theta,x=(b_i,e_i)]{s_1s_2}=\sum_{jk}q_{ij}h_k\fsem[\theta,x=(b_i,e_i),y=(b'_k,f_k)]{s'_{ij}}$.
    Therefore, we have
    $\fsem{(s_1s_2)[r/x]}=\sum_ip_i\fsem[\theta,x=(b_i,e_i)]{s_1s_2}$. Since
    $s'_{ij}$ is smaller than $s_1$, the induction hypothesis applies, and so,
    if $\sem[\theta,x=(b_i,e_i),y=(b'_k,f_k))]{s'_{ij}} =
    \{(q_{ij},b'_{ij},\rho_{ij})\}_j$, then
    $\sem[\theta,y=(b'_k,f_k)]{s'_{ij}[r/x]}
    =\{(p_iq_{ij},b'_{ij},\rho_{ij})\}_{ij}$.

    \newcommand\macro{\fsem[\theta,x=(b_i,e_i)]{s_2}}
  \item Let $x\in FV(s_s)$. By Lemma~\ref{lem:InterpAbs}, we have
    $\fsem[\theta,x=(b_i,e_i)]{s_1}=\sum_jq_j\para\mapsto\fsem[\theta,x=(b_i,e_i),y=\para]{s'_{ij}}$.
    By Corollary~\ref{cor:InterpApp}, we have that
    $\fsem[\theta,x=(b_i,e_i)]{s_1s_2}=
    \sum_jq_j\fsem[\theta,x=(b_i,e_i),y=(\emptybit,\macro)]{s'_{ij}}$. By the
    induction hypothesis, we have
    $\fsem{s_2[r/x]}=\sum_ip_i\fsem[\theta,x=(b_i,e_i)]{s_2}$. Since $x\notin
    FV(s_2[r/x])$, we have
    $\fsem{s_2[r/x]}=\fsem[\theta,x=(b_i,e_i)]{s_2[r/x]}$. Therefore, by
    Corollary~\ref{cor:InterpApp}, we have
    $\fsem[\theta,x=(b_i,e_i)]{s_1s_2[r/x]}=
    \sum_{ij}p_iq_j\fsem[\theta,x=(b_i,e_i),y=(\emptybit,\macro)]{s_{ij}'}$\parfillskip=0pt

    $=\sum_ip_i\fsem[\theta,x=(b_i,e_i)]{s_1s_2}$. Since $s'_{ij}$ is smaller
    than $s_1$, the induction hypothesis applies, and therefore, if
    \newcommand\aca{\fsem[\theta,x=(b_i,e_i)]{s_2}}
    $\sem[\theta,x=(b_i,e_i),y=(\emptybit,\aca)]{s'_{ij}}=\{(q_{ij},b'_{ij},\rho_{ij})\}_j$,
    then
    $\sem[\theta,y=(\emptybit,\aca)]{s'_{ij}[r/x]}=\{(p_iq_{ij},b'_{ij},\rho_{ij})\}_{ij}$.
  \end{itemize}
\item Let $t=\rho^n$.
  $\fsem{\rho^n[r/x]}=\fsem{\rho^n}=\rho^n=\fsem[\theta,x=(b_i,e_i)]{\rho^n}=\sum_ip_i\fsem[\theta,x=(b_i,e_i)]{\rho^n}$.
  In addition,
  $\sem[\theta,x=(b_i,e_i)]{\rho^n}=\{(1,\emptybit,\rho^n)\}=\{(p_i,\emptybit,\rho^n)\}_i=\sem{\rho^n[r/x]}$.
\item Let $t=U^m s$. By the induction hypothesis,
  $\fsem{s[r/x]}=\sum_ip_i\fsem[\theta,x=(b_i,e_i)] s$. By
  Lemma~\ref{lem:InterpUM}, $\fsem{U^ms[r/x]}
  =\ext{U^m}\fsem{s[r/x]}\ext{U^m}^\dagger
  =\ext{U^m}\sum_ip_i\fsem[\theta,x=(b_i,e_i)] s\ext{U^m}^\dagger $
  $=\sum_ip_i\ext{U^m}\fsem[\theta,x=(b_i,e_i)] s\ext{U^m}^\dagger
  =\sum_ip_i\fsem[\theta,x=(b_i,e_i)]{U^m s}$. Let
  $\sem[\theta,x=(b_i,e_i)]s=\{(q_{ij},b_{ij},\rho_{ij})\}_j$, then
  $\sem[\theta,x=(b_i,e_i)]{U^ms}=\{(q_{ij},\emptybit,\ext{U^m}\rho_{ij}\ext{U^m})\}_j$.
  In addition, by the induction hypothesis, we have
  $\sem{s[r/x]}=\{(p_iq_{ij},b_{ij},\rho_{ij})\}_{ij}$, therefore
  $\sem{(U^ms)[r/x]}=\sem{U^ms[r/x]}=\{(p_iq_{ij},\emptybit,\ext{U^m}\rho_{ij}\ext{U^m})\}_{ij}$.
\item Let $t=\pi^m s$. By the induction hypothesis,
  $\fsem{s[r/x]}=\sum_ip_i\fsem[\theta,x=(b_i,e_i)] s$. By
  Lemma~\ref{lem:InterpUM}, $\fsem{\pi^ms[r/x]}
  =\sum_j\ext{\pi_j}\fsem{s[r/x]}\ext{\pi_j}^\dagger
  =\!\sum_j\ext{\pi_j}\sum_jp_i\fsem[\theta,x=(b_i,e_i)]s\ext{\pi_j}^\dagger $
  $=\sum_ip_i\sum_j\ext{\pi_j}\fsem[\theta,x=(b_i,e_i)]s\ext{\pi_j}^\dagger
  =\sum_ip_i\fsem[\theta,x=(b_i,e_i)]{\pi^ms}$. Let
  $\sem[\theta,x=(b_i,e_i)]s=\{(q_{ij},b_{ij},\rho_{ij})\}_j$, then
  $\sem[\theta,x=(b_i,e_i)]{\pi^ms}=\{(q_{ij}\tr(\ext{\pi_k}^\dagger\ext{\pi_k}\rho_{ij}),k,\frac{\ext{\pi_k}\rho_{ij}\ext{\pi_k}}{\tr(\ext{\pi_k}^\dagger\ext{\pi_k}\rho_{ij})})\}_{jk}$.
  By the induction hypothesis,
  $\sem{s[r/x]}=\{(p_iq_{ij},b_{ij},\rho_{ij})\}_{ij}$, therefore, we have
  $\sem{(\pi^ms)[r/x]}=\sem{\pi^ms[r/x]}=\{(p_iq_{ij}\tr(\ext{\pi_k}^\dagger\ext{\pi_k}\rho_{ij}),k,\frac{\ext{\pi_k}\rho_{ij}\ext{\pi_k}}{\tr(\ext{\pi_k}^\dagger\ext{\pi_k}\rho_{ij})})\}_{ijk}$.
\item Let $t=s_1\otimes s_2$, with $x\in FV(s_1)$. By the induction hypothesis,
  $\fsem{s_1[r/x]}=\sum_ip_i\fsem[\theta,x=(b_i,e_i)]{s_1}$, and notice that,
  since $x\notin FV(s_2)$, we have $\fsem{s_2}=\fsem[\theta,x=(b_i,e_i)]{s_2}$.
  Therefore, by Lemma~\ref{lem:InterpTensor}, $\fsem{s_1[r/x]\otimes s_2}\!
  =\!\fsem{s_1[r/x]}\otimes\fsem{s_2}
  =(\sum_ip_i\fsem[\theta,x=(b_i,e_i)]{s_1})\otimes\fsem[\theta,x=(b_i,e_i)]{s_2}
  \!=\!\sum_ip_i\fsem[\theta,x=(b_i,e_i)]{s_1}\otimes\fsem[\theta,x=(b_i,e_i)]{s_2}
  $ $=\sum_ip_i\fsem[\theta,x=(b_i,e_i)]{s_1\otimes s_2}$, where the last step
  is due to Lemma~\ref{lem:InterpTensor} as well. Note that if
  $\sem[\theta,x=(b_i,e_i)]{s_1}=\{(q_{ij},b_{ij},\rho_{ij})\}_j$, and
  $\sem[\theta,x=(b_i,e_i)]{s_2}=\{(h_k,b'_k,\rho'_k)\}_k$, we have
  $\sem[\theta,x=(b_i,e_i)]{s_1\otimes s_i} =
  \{(q_{ij}h_k,\emptybit,\rho_{ij}\otimes\rho'_k)\}_{kj}$. By the induction
  hypothesis $\sem{s_1[r/x]}=\{(p_iq_{ij},b_{ij},\rho_{ij})\}_{ij}$, and so
  $\sem{(s_1\otimes s_2)[r/x]}=\sem{s_1[r/x]\otimes s_2} =
  \{(p_iq_{ij}h_k,\emptybit,\rho_{ij}\otimes\rho'_k)\}_{ijk}$.
\item Let $t=(b^m,\rho^n)$.
  $\fsem{(b^m,\rho^n)[r/x]}=\fsem{(b^m,\rho^n)}=\rho^n=\fsem[\theta,x=(b_i,e_i)]{(b^m,\rho^n)}$
  $=\sum_ip_i\fsem[\theta,x=(b_i,e_i)]{(b^m,\rho^n)}$. Moreover,
  $\sem[\theta,x=(b_i,e_i)]{(b^m,\rho^n)}
  =\{(1,b^m,\rho^n)\}=\{(p_i,b^m,\rho^n)\}_i
  =\sem{(b^m,\rho^n[r/x])}=\sem{(b^m,\rho^n)[r/x]}$.
\item Let $t=\letcase ys{t_0,\dots,t_{2^m-1}}$. Let $\sem[\theta,x=(b_i,e_i)] s
  = \{(q_{ij},b'_{ij},\rho_{ij})\}_j$. Then, by the induction hypothesis,
  $\sem{s[r/x]} = \{p_iq_{ij},b'_{ij},\rho_{ij}\}_{ij}$. Let, forall $i$ and
  $j$, $\sem[\theta,x = (b_i,e_i),y = (\emptybit,\rho_{ij})]{t_{b'_{ij}}} =
  \{(h_{ijk},b''_{ijk},f_{ijk})\}_k$. Hence, we have that $\sem[\theta,x =
  (b_i,e_i)]{\letcase ys{t_0,\dots,t_{2^m-1}}} =
  \{(q_{ij}h_{ijk},b''_{ijk},f_{ijk}\}_{jk}$ and also $\sem{\letcase
    y{s[r/x]}{t_0,\dots,t_{2^m-1}}}
  =\{(p_iq_{ij}h_{ijk},b''_{ijk},f_{ijk})\}_{ijk}$. Therefore, we have
  $\fsem{\letcase y{s[r/x]}{t_0,\dots,t_{2^m-1}}} =
  \sum_ip_i\fsem[\theta,x=(b_i,e_i)]{\letcase ys{t_0,\dots,t_{2^m-1}}}$ \qed
\end{itemize}

\subsection{Proof of Theorem~\ref{thm:IntRed}}\label{proof:IntRed}
By induction on the relation $\topr[p]$.
\begin{itemize}
\item $(\lambda x.t)r\topr t[r/x]$. We have $\sem{\lambda
    x.t}=\{1,\emptybit,\para\mapsto\sem[\theta,x=\para]t)\}$. By
  Lemma~\ref{lem:wellDefined}, we have $\sem r=\{(p_i,b_i,e_i)\}_i$ with $\sum_i
  p_i=1$, and $\sem[\theta,x=(b_i,e_i)]t=\{(q_{ij},b'_{ij},g_{ij})\}_j$.
  Therefore, we have that $\sem{(\lambda
    x.t)r}=\{(p_iq_{ij},b'_{ij},g_{ij})\}_{ij}$, and $\fsem{(\lambda
    x.t)r}=\sum_{ij}p_iq_{ij}g_{ij}$ which, by Lemma~\ref{lem:substitutionSem},
  is equal to $\fsem{t[r/x]}$.
\item $\letcase x{(b^m,\rho^n)}{t_0,\dots,t_{2^m-1}}\topr t_{b^m}[\rho^n/x]$.
  Then, we have that $\sem{\letcase x{(b^m,\rho^n)}{t_0,\dots,t_{2^m-1}}} =
  \sem[\theta,x=(\emptybit,\rho^n)]{t_{b^m}}$. Therefore, we have that
  $\fsem{\letcase x {(b^m,\rho^n)} {t_0,\dots,t_{2^m-1}}} =
  \fsem[\theta,x=(\emptybit,\rho^n)]{t_{b^m}}$. Since
  $\sem{\rho^n}=\{(1,\emptybit,\rho^n)\}$, we have, by
  Lemma~\ref{lem:substitutionSem},
  $\fsem[\theta,x=(\emptybit,\rho^n)]{t_{b^m}}=\fsem{t_{b^m}[\rho^n/x]}$.
\item $U^m\rho^n\topr\rho'^n$, with $\rho'^n=\ext{U^m}\rho^n\ext{U^m}^\dagger$.
  Then $\sem{U^m\rho^n}=\{(1,\emptybit,\ext{U^m}\rho^n\ext{U^m}^\dagger)\}$
  $=\sem{\rho'^n}$, so $\fsem{U^m\rho^n}=\rho'^n=\fsem{\rho'^n}$.
\item $\pi^n\rho^n\topr[p_i](i^m,\rho_i^n)$, with
  $p_i=\tr(\ext{\pi_i}^\dagger\ext{\pi_i}\rho^n)$ and
  $\rho_i^n=\frac{\ext{\pi_i}\rho^n\ext{\pi_i}^\dagger}{p_i}$. We have
  $\sem{\pi^m\rho^n}=\{(p_i,i,\rho_i^n)\}_i$. Hence,
  $\fsem{\pi^m\rho^n}=\sum_ip_i\rho_i^n=\sum_i p_i\fsem{\rho_i^n}$.
\item $\rho_1\otimes\rho_2\topr\rho$, with $\rho=\rho_1\otimes\rho_2$. We have
  $\sem{\rho_1\otimes\rho_2}=\{(1,\emptybit,\rho)\}$, so
  $\fsem{\rho_1\otimes\rho_2}=\rho=\fsem{\rho}$.
\item Contextual cases: Let $s\topr[p_i] s_i$. Then
  \begin{itemize}
  \item $\lambda x.s\topr[p_i]\lambda x.s_i$. By the induction hypothesis,
    $\fsem[\theta']s=\sum_ip_i\fsem[\theta']{s_i}$, for all
    $\theta'\vDash\Gamma,x:A$. Then, $\fsem{\lambda
      x.s}=\para\mapsto\fsem[\theta,x=\para]s=\para\mapsto
    \sum_ip_i\fsem[\theta,x=\para]{s_i} =\sum_i
    p_i(\para\mapsto\fsem[\theta,x=\para]{s_i}) =\sum_ip_i\fsem{\lambda x.s_i}$.
  \item $ts\topr[p_i]ts_i$. By the induction hypothesis, $\fsem s=\sum_ip_i\fsem
    {s_i}$. Then, we have
    $\sem{ts}=\{(p_iq_jh_{ijk},b''_{ijk},g_{ijk})\}_{ijk}$, with $\sem
    s=\{(p_i,b_i,e_i)\}_i$, $\sem t=\{(q_j,b'_j,f_j)\}_j$, and
    $f_j(b_i,e_i)=\{(h_{ijk},b''_{ijk},g_{ijk})\}_k$. Hence,
    $\fsem{ts}=\sum_{ijk}p_iq_jh_{ijk}g_{ijk}
    =\sum_ip_i(\sum_{jk}q_jh_{ijk}g_{ijk}) =\sum_ip_i\fsem{ts_i}$
  \item $st\topr[p_i]s_it$. By the induction hypothesis, $\fsem s=\sum_ip_i\fsem
    {s_i}$. Then, we have
    $\sem{st}=\{(q_jp_ih_{jik},b''_{jik},g_{jik})\}_{jik}$, with $\sem
    t=\{(q_j,b'_j,f_j)\}_j$, $\sem s=\{(p_i,b_i,e_i)\}_i$, and
    $e_i(b_j,f_j)=\{(h_{jik},b''_{jik},g_{jik})\}_k$. Hence,
    $\fsem{st}=\sum_{jik}q_jp_ih_{jik}g_{jik}
    =\sum_ip_i(\sum_{jk}q_jh_{ijk}g_{ijk}) =\sum_ip_i\fsem{s_it}$
  \item $U^ms\topr[p_i] U^ms_i$. By the induction hypothesis, $\fsem
    s=\sum_ip_i\fsem {s_i}$. By Lemma~\ref{lem:InterpUM}, we have $\fsem{U^m s}
    =\ext{U^m}\sum_ip_i\fsem{s_i}\ext{U^m}^\dagger
    =\sum_ip_i\ext{U^m}\fsem{s_i}\ext{U^m}^\dagger$ $=\sum_ip_i\fsem{U^ms_i}$.
  \item $\pi^ms\topr[p_i]r=\pi^ms_i$. By the induction hypothesis, $\fsem
    s=\sum_ip_i\fsem {s_i}$. By Lemma~\ref{lem:InterpUM}, we have $\fsem{\pi^m
      s} =\ext{\pi_j}\sum_ip_i\fsem{s_i}\ext{\pi_j}^\dagger
    =\sum_ip_i\ext{\pi_j}\fsem{s_i}\ext{\pi_j}^\dagger
    =\sum_ip_i\fsem{\pi^ms_i}$.
  \item $t\otimes s\topr[p_i]r=t\otimes s_i$. By the induction hypothesis,
    $\fsem s=\sum_ip_i\fsem {s_i}$. By Lemma~\ref{lem:InterpTensor},
    $\fsem{t\otimes s}=\fsem t\otimes\fsem s=\fsem t\otimes(\sum_ip_i\fsem{s_i})
    =\sum_ip_i\fsem t\otimes\fsem{s_i}=\sum_ip_i\fsem{t\otimes s_i}$.
  \item $s\otimes t\topr[p_i]r=s_i\otimes t$. Analogous to previous case.
  \item $\letcase xs{t_0,\dots,t_{2^m-1}}\topr[p_i]\letcase
    x{s_i}{t_0,\dots,t_{2^m-1}}$. By the induction hypothesis, we have that
    $\fsem s=\sum_ip_i\fsem {s_i}$. Let $\sem s=\{(p_i,b_i,\rho_i)\}_i$ and
    $\sem[\theta,x=(\emptybit,\rho_i)]{t_{b_i}}=\{(q_{ij},b'_{ij},e_{ij})\}_j$,
    then, $\fsem{\letcase
      xs{t_0,\dots,t_{2^m-1}}}=\sum_{ij}p_iq_{ij}e_{ij}=\sum_ip_i(\sum_jq_{ij}e_{ij})$
    Notice that $\sum_jq_{ij}e_{ij}=\fsem{\letcase
      x{s_i}{t_0,\dots,t_{2^m-1}}}$. \qed
  \end{itemize}
\end{itemize}
\section{Detailed Proofs of Section~\ref{sec:ExtDenSem}}
\subsection{Proof of Theorem~\ref{thm:ExtIntRed}}\label{proof:ExtIntRed}
By induction on the relation $\toD$. Rules $(\lambda x.t)r\toD t[r/x]$,
$U^m\rho^n\toD \rho'$ and $\rho\otimes\rho'\toD\rho''$ are also valid rules for
relation $\topr[1]$, and hence the proof of these cases are the same than in
Theorem~\ref{thm:IntRed}.

Remaining cases:
\begin{itemize}
\item $\qletcase x{\pi^m\rho^n}{t_0,\dots,t_{2^m-1}} \toD \sum_i
  p_it_i[\rho'/x]$, with $p_i=\tr(\ext{\pi_i}^\dagger\ext{\pi_i}\rho^n)$ and
  $\rho' = \frac{\ext{\pi_i}\rho^n\ext{\pi_i}^\dagger}{p_i}$. Since the
  interpretation of $\mathsf{letcase}^\circ$ coincides with the interpretation
  of $\mathsf{letcase}$, and $\letcase x{\pi^m\rho^n}{t_0,\dots,t_{2^m-1}}
  \topr[p_i] t_i[\rho'/x]$, we can conclude by Theorem~\ref{thm:IntRed}, and
  Lemma~\ref{lem:ExtSumFsem} that $\fsem{\qletcase x{\pi^m\rho^n}
    {t_0,\dots,t_{2^m-1}}}
  =\sum_ip_i\fsem{t_i[\rho'/x]}=\fsem{\sum_ip_i{t_i[\rho'/x]}}$.
\item $\sum_i p_i\rho_i\toD \rho'$, with $\rho'=\sum_ip_i\rho_i$.
  $\fsem{\sum_ip_i\rho_i}=\sum_ip_i\rho_i=\fsem{\rho'}$.
\item $\sum_i p_i t \toD t$. Let $\sem t=\{(q_j,b_j,e_j)\}_j$. Then,
  $\sem{\sum_ip_it}=\{(p_iq_j,b_j,e_j)\}_{ij}$ and so,
  $\fsem{\sum_ip_it}=\sum_{ij}p_iq_je_j=(\sum_ip_i)\sum_jq_je_j=\sum_jq_je_j=\fsem
  t$.
\item $(\sum_i p_i t_i)r \toD \sum_i p_i (t_ir)$. Let
  $\sem{t_i}=\{(q_{ij},b_{ij},f_{ij})\}_j$, $\sem{r}=\{(h_k,b'_k,e_k)\}_k$, and
  $f_{ij}(b'_k,e_k)=\{(\ell_{ijkh},b''_{ijkh},g_{ijkh})\}_h$. Then, we have that
  $\sem{t_ir}$ is equal to
  $\{(q_{ij}h_k\ell_{ijkh},b''_{ijkh},g_{ijkh})\}_{jkh}$, and
  $\fsem{\sum_ip_i(t_ir)}=\sum_{jkh}p_iq_{ij}h_k\ell_{ijkh}g_{ijkh}$. On the
  other hand, $\sem{\sum_ip_it_i}=\{(p_iq_{ij},b_{ij},f_{ij})\}_{ij}$, and so
  $\sem{(\sum_ip_it_i)r}=\{\!(p_iq_{ij}h_k\ell_{ijkh},b''_{ijkh},g_{ijkh})\!\}_{ijkh}$.
  Hence,
  $\fsem{(\sum_i\!p_it_i)r}\!=\!\sum_{jkh}p_iq_{ij}h_k\ell_{ijkh}g_{ijkh}$
  $=\fsem{\sum_ip_i(t_ir)}$.

\item Contextual cases: Only two cases need to be checked. All the other cases
  are analogous to those of the proof of Theorem~\ref{thm:IntRed}.
  \begin{itemize}
  \item Let $\sum_ip_it_i\toD\sum_ip_ir_i$, where $t_j\toD r_j$ and $\forall
    i\neq j$, $t_i=r_i$. By the induction hypothesis, $\fsem{t_j}=\fsem{r_j}$.
    Hence, using Lemma~\ref{lem:ExtSumFsem},
    $\fsem{\sum_ip_it_i}=\sum_ip_i\fsem{t_i}=\sum_ip_i\fsem{r_i}=\fsem{\sum_ip_ir_i}$.
  \item Let $\qletcase{x}{t}{s_0,\dots,s_n}\toD\qletcase{x}{r}{s_0,\dots,s_n}$,
    where $t\toD r$. Let $\sem t=\{(p_i,b_i,\rho_i)\}_i$ and
    $\sem[\theta,x=(\emptybit,\rho_i)]{s_{b_i}}=\{(q_{ij},b'_{ij},e_{ij})\}_j$,
    then $\fsem{\qletcase{x}{t}{s_0,\dots,s_{2^m-1}}}=\sum_{ij}p_iq_{ij}e_{ij}$.

    On the other hand, let $\sem r=\{(h_k,b''_k,\rho'_k)\}_k$. By the induction
    hypothesis, $\sum_kh_k\rho'_k=\sum_ip_i\rho_i$. We conclude by inversion
    that $h_i=p_i$ and $\rho'_i=\rho_i$, which prove the case. \qed
  \end{itemize}
\end{itemize}
\section{Detailed Proofs of Section~\ref{sec:correctness}}
\subsection{Proof of Lemma~\ref{lem:substitution}}\label{proof:substitution}
By induction on $t$. The only difference between $\lambdens$ and $\qlambdens$
are in terms $(b^m,\rho^n)$, which is not present in $\qlambdens$, and
$\sum_ip_it_i$, which is not present in $\lambdens$. Hence, we can prove both
calculi at the same time. We use $\vdash$ generically to refer also to $\Vdash$
where it is also valid.
\begin{itemize}
\item Let $t=x$. Then $B=A$. By Lemma~\ref{lem:weak}, $\Gamma,\Delta\vdash r:A$.
  Notice that $t[r/x]=r$.
\item Let $t=y$. Then, by Lemmas~\ref{lem:weak} and~\ref{lem:streng},
  $\Gamma,\Delta\vdash y:B$. Notice that $t[r/x]=y$.
\item Let $t=\lambda y.s$. Then $B=C\multimap D$ and, by inversion,
  $\Gamma,x:A,y:C\vdash s:D$. Then, by the induction hypothesis,
  $\Gamma,y:C,\Delta\vdash s[r/x]:D$, so, by rule $\multimap_i$,
  $\Gamma,\Delta\vdash\lambda y.(s[r/x]):C\multimap D$. Notice that $\lambda
  y.(s[r/x])=(\lambda y.s)[r/x]$.
\item Let $t=t_1t_2$. Then $\Gamma,x:A=\Gamma_1,\Gamma_2$, with $\Gamma_1\vdash
  t_1:C\multimap B$ and $\Gamma_2\vdash t_2:C$.
  \begin{itemize}
  \item If $x:A\in\Gamma_1$, then, by the induction hypothesis
    $\Gamma_1\setminus\{x:A\},\Delta\vdash t_1[r/x]:C\multimap B$, so by rule
    $\multimap_e$, $\Gamma_1\setminus\{x:A\},\Gamma_2,\Delta\vdash
    t_1[r/x]t_2:B$. Notice that $\Gamma_1\setminus\{x:A\},\Gamma_2=\Gamma$ and
    $t_1[r/x]t_2=(t_1t_2)[r/x]$.
  \item If $x:A\in\Gamma_2$, then, by the induction hypothesis
    $\Gamma_2\setminus\{x:A\},\Delta\vdash t_2[r/x]:C$, so by rule
    $\multimap_e$, $\Gamma_1,\Gamma_2\setminus\{x:A\},\Delta\vdash
    t_1(t_2[r/x]):B$. Notice that $\Gamma_1,\Gamma_2\setminus\{x:A\}=\Gamma$ and
    $t_1(t_2[r/x])=(t_1t_2)[r/x]$.
  \end{itemize}
\item Let $t=\rho^n$. Then $B=n$. By Lemmas~\ref{lem:weak} and~\ref{lem:streng},
  $\Gamma,\Delta\vdash\rho^n:n$. Notice that $t[r/x]=\rho^n$.
\item Let $t=U^ms$. Then $B=n$ and $\Gamma,x:A\vdash s:n$. Then, by the
  induction hypothesis, $\Gamma,\Delta\vdash s[r/x]:n$. So, by rule $u$,
  $\Gamma,\Delta\vdash U^m(s[r/x]):n$. Notice that $U^m(s[r/x])=(U^ms)[r/x]$.
\item Let $t=\pi^ms$. Then $B=(m,n)$ and $\Gamma,x:A\vdash s:n$. Then, by the
  induction hypothesis, $\Gamma,\Delta\vdash s[r/x]:n$. So, by rule $m$,
  $\Gamma,\Delta\vdash \pi^m(s[r/x]):(m,n)$. Notice that
  $\pi^m(s[r/x])=(\pi^ms)[r/x]$.
\item Let $t=t_1\otimes t_2$. Then $B=n_1+n_2$, $\Gamma,x:A=\Gamma_1,\Gamma_1$
  with $\Gamma_i\vdash t_i:n_i$ for $i=1,2$. Let $x:A\in\Gamma_i$ for some
  $i=1,2$. Then, by the induction hypothesis,
  $\Gamma_i\setminus\{x:A\},\Delta\vdash t_i[r/x]$, so by rule $\otimes$, either
  $\Gamma,\Delta\vdash t_1[r/x]\otimes t_2:n_1+n_2$, or $\Gamma,\Delta\vdash
  t_1\otimes t_2[r/x]:n_1+n_2$. In the first case, notice that $t_1[r/x]\otimes
  t_2=(t_1\otimes t_2)[r/x]$, and in the second, $t_1\otimes
  t_2[r/x]=(t_1\otimes t_2)[r/x]$.
\item Let $t=(b^m,\rho^n)$. Then $B=(m,n)$. By Lemmas~\ref{lem:weak}
  and~\ref{lem:streng}, $\Gamma,\Delta\vdash(b^m,\rho^n):(m,n)$. Notice that
  $t[r/x]=(b^m,\rho^n)$.
\item Let $t=\sum_ip_it_i$. Then $\Gamma,x:A\Vdash t_i:B$ and so, by the
  induction hypothesis, $\Gamma,\Delta\Vdash t_i[r/x]:B$. Therefore, by rule
  $+$, $\Gamma,\Delta\Vdash\sum_ip_it_i[r/x]:B$.
\item Let $t=\letcase ys{t_0,\dots,t_{2^m-1}}$. $y:n\vdash t_i:B$, for
  $i=0,\dots,2^m-1$, and, $\Gamma\vdash s:(m,n)$. By the induction hypothesis,
  $\Gamma,\Delta\vdash s[r/x]:(m,n)$. So, by rule $\mathsf{lc}$,
  $\Gamma,\Delta\vdash\letcase y{s[r/x]}{t_0,\dots,t_{2^m-1}}\!\!:B$. Notice
  that $\letcase y{s[r/x]}{t_0,\dots,t_{2^m-1}}=(\letcase
  ys{t_0,\dots,t_{2^m-1}})[r/x]$.\qed
\end{itemize}

\subsection{Proof of Theorem~\ref{thm:SR}}\label{proof:SR}
By induction on the relations $\topr[p]$ and $\toD$.
\begin{enumerate}
\item
  \begin{itemize}
  \item Let $t=(\lambda x.t')s$, $r=t'[r/s]$ and $p=1$. Then
    $\Gamma\vdash(\lambda x.t')s:A$, so, $\Gamma_1\vdash\lambda x.t':B\multimap
    A$ and $\Gamma_2\vdash s:B$, with $\Gamma=\Gamma_1,\Gamma_2$. Hence,
    $\Gamma_1,x:B\vdash t':A$, and so, by Lemma~\ref{lem:substitution},
    $\Gamma\vdash t'[s/x]:A$.
  \item Let $t=\letcase x{(b^m,\rho^n)}{t_0,\dots,t_{2^m-1}}$,
    $r=t_{b^m}[\rho^n/x]$ and $p=1$. Then, by inversion, $x:n'\vdash t_i:A$, for
    $i=0,\dots,2^m-1$, and $\Gamma\vdash
    (b^m,\rho^n):(m',n')$. 
    So, by inversion again, $m'=m$ and $n'=n$. By rule $\mathsf{ax}_\rho$,
    $\Gamma\vdash\rho^n:n$. Hence, by Lemma~\ref{lem:substitution},
    $\Gamma\vdash t_{b^m}[\rho^n/x]:A$.
  \item Let $t=U^m\rho^n$, $r={\rho'}^n$ and $p=1$, with
    ${\rho'}^n=\ext{U^m}\rho^n\ext{U^m}^\dagger$. Then $A=n$. By rule
    $\mathsf{ax}_\rho$, $\Gamma\vdash{\rho'}^n:n$.
  \item Let $t=\pi^m\rho^n$, $r=(i^m,\rho_i^n)$, with
    $\rho_i^n=\frac{\ext{\pi_i}\rho^n\ext{\pi_i}^\dagger}p$ and
    $p=\tr(\ext{\pi_1}^\dagger\ext{\pi_i}\rho^n)$. Then $B=(m,n)$. By rule
    $\mathsf{m}$, $\Gamma\vdash (i^m,\rho_i^n):(m,n)$.
  \item Let $t=\rho_1^n\otimes\rho_2^m$ and $r=\rho$, with
    $\rho=\rho_1^n\otimes\rho_2^m$. Then, $A=n+m$, with $\vdash\rho_1^n:n$ and
    $\vdash\rho_2^m:m$. Since $\rho$ is a density matrix of $(n+m)$-qubits,
    $\vdash\rho:n+m$.
  \item Contextual cases: Let $s\topr[p] s'$, then
    \begin{itemize}
    \item Consider $t=\lambda x.s$ and $r=\lambda x.s'$. Then $A=B\multimap C$
      and $\Gamma,x:B\vdash s:C$. So, by the induction hypothesis,
      $\Gamma,x:B\vdash s':C$ and by rule $\multimap_i$, $\Gamma\vdash \lambda
      x.s':B\multimap C$.
    \item Consider $t=t's$ and $r=t's'$. Then $\Gamma=\Gamma_1,\Gamma_2$, with
      $\Gamma_1\vdash t':B\multimap A$ and $\Gamma_2\vdash s:B$. By the
      induction hypothesis, $\Gamma_2\vdash s':B$, so by rule $\multimap_e$,
      $\Gamma\vdash t's':A$.
    \item Consider $t=st'$ and $r=s't'$. Then $\Gamma=\Gamma_1,\Gamma_2$, with
      $\Gamma_1\vdash s:B\multimap A$ and $\Gamma_2\vdash t':B$. By the
      induction hypothesis, $\Gamma_1\vdash s':B\multimap A$, so by rule
      $\multimap_e$, $\Gamma\vdash s't':A$.
    \item Consider $t=U^ms$ and $r=U^ms'$. Then $A=n$ and $\Gamma\vdash s:n$. By
      the induction hypothesis $\Gamma\vdash s':n$, so by rule $\mathsf{u}$,
      $\Gamma\vdash U^ms':n$.
    \item Consider $t=\pi^ms$ and $r=\pi^ms'$. Then $A=n$ and $\Gamma\vdash
      s:n$. By the induction hypothesis $\Gamma\vdash s':n$, so by rule
      $\mathsf{m}$, $\Gamma\vdash \pi^ms':n$.
    \item Consider $t=t'\otimes s$ and $r=t'\otimes s'$. Then $A=n+m$ and
      $\Gamma=\Gamma_1,\Gamma_2$, with $\Gamma_1\vdash t':n$ and $\Gamma_2\vdash
      s:m$. By the induction hypothesis $\Gamma_2\vdash s':m$, so by rule
      $\otimes$, $\Gamma\vdash t'\otimes s':n+m$.
    \item Consider $t=s\otimes t'$ and $r=s'\otimes t'$. Then $A=n+m$ and
      $\Gamma=\Gamma_1,\Gamma_2$, with $\Gamma_1\vdash s:n$ and $\Gamma_2\vdash
      t':m$. By the induction hypothesis $\Gamma_2\vdash s':m$, so by rule
      $\otimes$, $\Gamma\vdash t'\otimes s':n+m$.
    \item Consider $t=\letcase xs{t_0,\dots,t_{2^m-1}}$ and $r=\letcase
      x{s'}{t_0,\dots,t_{2^m-1}}$. Then $x:n\vdash t_i:A$ for $i=0,\dots,2^m-1$,
      and $\Gamma\vdash s:(m,n)$. So, by the induction hypothesis, $\Gamma\vdash
      s':(m,n)$ and by rule $\mathsf{lc}$, $\Gamma\vdash\letcase
      x{s'}{t_0,\dots,t_{2^m-1}}:A$.\qed
    \end{itemize}
  \end{itemize}
\item
  \begin{itemize}
  \item Let $t=(\lambda x.t')s$ and $r=t'[s/x]$. Analogous to the same rule in
    $\lambdens$.
  \item Let $t=\qletcase x{\pi^m\rho^n}{t_0,\dots,t_{2^m}-1}$ and $r=\sum_i p_i
    t_i[\rho^n_i/x]$, with
    $\rho_i^n=\frac{\ext{\pi_i}\rho^n\ext{\pi_i}^\dagger}{p_i}$ and
    $p_i=\tr(\ext{\pi_i}^\dagger\ext{\pi_i}\rho^n)$. Then
    $\Gamma\Vdash\pi^m\rho^n:(m,n)$ and $x:n\Vdash t_i:A$. By
    Lemma~\ref{lem:substitution}, $\Gamma\Vdash t_i[\rho_i^n/x]:A$, then, by
    rule $+$, $\Gamma\Vdash\sum_ip_it_i[\rho_i^n/x]:A$.
  \item Let $t=U^m\rho^n$ and $r={\rho'}^n$, with
    $\ext{U^m}\rho^n\ext{U^m}^\dagger={\rho'}^n$. Analogous to the same rule in
    $\lambdens$.
  \item Let $t=\rho_1^n\otimes\rho^m_2$ and $r=\rho$, with
    $\rho=\rho_1^n\otimes\rho_2^m$. Analogous to the same rule in $\lambdens$.
  \item Let $t=\sum_i p_i\rho_i$ and $r=\rho'$, with $\rho'=\sum_ip_i\rho_i$.
    Then, $\Gamma\Vdash\sum_ip_i\rho_i:n$, and by rule $\mathsf{ax}_\rho$,
    $\Gamma\Vdash\rho':n$.
  \item Let $t=\sum_i p_i r$. Then, $\Gamma\Vdash r:A$.
  \item Let $t=(\sum_i p_i t_i)r$ and $\sum_i p_i (t_ir)$. Then,
    $\Gamma=\Gamma_1,\Gamma_2$, $\Gamma_1\Vdash t_i:B\multimap A$ and
    $\Gamma_2\Vdash r:B$. Therefore, by rule $\multimap_e$,
    $\Gamma_1,\Gamma_2\Vdash t_ir:A$, and by rule $+$,
    $\Gamma_1,\Gamma_2\Vdash\sum_ip_i(t_ir):A$.
  \item Contextual cases. All the contextual cases are analogous to the same
    rules in $\lambdens$, except for the contextual rule of $\sum$: Consider
    $t=\sum_ip_it_i$ and $r=\sum_ip_ir_i$, with $t_j\toD r_j$, and $\forall
    i\neq j$, $t_i=r_i$. By inversion, $\Gamma\Vdash t_i:A$. By the induction
    hypothesis, $\forall i$, $\Gamma\Vdash r_i:A$. Then, by rule $+$,
    $\Gamma\Vdash\sum_ip_ir_i$. \qed
  \end{itemize}
\end{enumerate}

\subsection{Proof of Lemma~\ref{lem:values}}\label{proof:values}
\begin{enumerate}
\item We proceed by induction on $v$.
  \begin{itemize}
  \item Let $v=x$. Then we are done, since variables do not rewrite.
  \item Let $v=\lambda x.v'$. By the induction hypothesis, $v'$ does not
    rewrite, and so $v$ neither does.
  \item Let $v=w_1\otimes w_2$. By the induction hypothesis, nor $w_1$ nor $w_2$
    rewrite, and the only rule rewriting a $\otimes$ in head position needs both
    $w_1$ and $w_2$ to be density matrices, which are not according to the
    grammar.
  \item Let $v=\rho^n$. Then we are done, since density matrices are constants
    and do not rewrite.
  \item Let $v=(b^m,\rho^n)$. Then we are done, since pairs are constants and do
    not rewrite.
  \end{itemize}
\item We proceed by induction on $v$. The only difference with the previous case
  is that the last term is not present, and a new term is introduced. Hence, let
  $v=\sum_i p_i w_i$. The only possible rewrite would be if there were some
  $w_i=w_j$, which is explicitly excluded from the grammar. So, this term does
  not rewrite \qed
\end{enumerate}

\subsection{Proof of Theorem~\ref{thm:Progress}}\label{proof:Progress}
We relax the hypotheses and prove the theorem for open terms as well. That is:
\begin{enumerate}
\item If $\Gamma\vdash t:A$, then either $t$ is a value, there exist $n$,
  $p_1,\dots,p_n$, and $r_1,\dots,r_n$ such that $t\topr[p_i]r_i$, or $t$
  contains a free variable, and $t$ does not rewrite.
\item If $\Gamma\Vdash t:A$, then either $t$ is a \qvalue, there exists $r$ such
  that $t\toD r$, or $t$ contains a free variable, and $t$ does not rewrite.
\end{enumerate}
The proofs proceeds as follows.
\begin{enumerate}
\item We proceed by induction on the derivation of $\vdash t:A$.
  \begin{itemize}
  \item Let $\Gamma,x:A\vdash x:A$ as a consequence of rule $\mathsf{ax}$. Then,
    we are done since $x$ is a free variable and does not rewrite.
  \item Let $\Gamma\vdash\lambda x.t:A\multimap B$ as a consequence of
    $\Gamma,x:A\vdash t:B$ and rule $\multimap_i$. Then, by the induction
    hypothesis, either
    \begin{itemize}
    \item $t$ is not a value and there exist $n$, $p_1,\dots,p_n$ and
      $r_1,\dots,r_n$ such that $t\topr[p_i] r_i$, which case $\lambda
      x.t\topr[p_i]\lambda x.r_i$;
    \item $t$ is a value, in which case $\lambda x.t$ is a value as well; or
    \item $t$ contains a free variable, and $t$ does not rewrite, in which case
      the same happens to $\lambda x.t$.
    \end{itemize}
  \item Let $\Gamma,\Delta\vdash tr:B$ as a consequence of $\Gamma\vdash
    t:A\multimap B$, $\Delta\vdash r:A$ and rule $\multimap_e$. Then, by the
    induction hypothesis, one of the following cases happens:
    \begin{itemize}
    \item There exist $n$, $p_1,\dots,p_n$, and $t_1,\dots,t_n$ such that
      $t\topr[p_i] t_i$, in which case $tr\topr[p_i]t_ir$.
    \item There exist $n$, $p_1,\dots,p_n$ and $r_1,\dots,r_n$ such that
      $r\topr[p_i] r_i$, in which case $tr\topr[p_i]tr_i$.
    \item $t$ is a value and $r$ does not rewrite. The only values which can be
      typed by $A\multimap B$ are:
      \begin{itemize}
      \item $t=x$, in which case $xr$ contains a free variable and does not
        rewrite.
      \item $t=\lambda x.v$, in which case $(\lambda x.v)r\topr v[r/x]$.
      \end{itemize}
    \item $t$ is not a value, contains a free variable, and does not rewrite,
      and $r$ does not rewrite, in which case $tr$ contains a free variable and
      does not rewrite.
    \end{itemize}
  \item Let $\Gamma\vdash\rho^n:n$ as a consequence of rule $\mathsf{ax_\rho}$.
    Then, we are done since $\rho^n$ is a value.
  \item Let $\Gamma\vdash U^mt:n$ as a consequence of $\Gamma\vdash t:n$ and
    rule $\mathsf{u}$. Then, by the induction hypothesis, one of the following
    cases happens:
    \begin{itemize}
    \item $t$ is a value. The only values which can be typed by $n$ are:
      \begin{itemize}
      \item $t=x$, in which case $U^mx$ contains a free variable and does not
        rewrite.
      \item $t=\bigotimes x_i$, in which case $U^m\bigotimes x_i$ contains free
        variables and does not rewrite.
      \item $t=\rho^n$, in which case $U^m\rho^n\topr[1]\rho'$, with
        $\rho'=\ext{U^m}\rho^n\ext{U^m}^\dagger$.
      \end{itemize}
    \item There exist $n$, $p_1,\dots,p_n$ and $t_1,\dots,t_n$ such that
      $t\topr[p_i]t_i$, in which case $U^mt\topr[p_i]t_i$;
    \item $t$ contains a free variable and does not rewrite, in which case the
      same is true for $U^mt$.
    \end{itemize}
  \item Let $\Gamma\vdash \pi^mt:(m,n)$ as a consequence of $\Gamma\vdash t:n$
    and rule $\mathsf{m}$. Then, by the induction hypothesis, one of the
    following cases happens:
    \begin{itemize}
    \item $t$ is a value. The only values which can be typed by $n$ are:
      \begin{itemize}
      \item $t=x$, in which case $\pi^mx$ contains a free variable and does not
        rewrite.
      \item $t=\bigotimes x_i$, in which case $\pi^m \bigotimes x_i$ contains
        free variables and does not rewrite.
      \item $t=\rho^n$, in which case $\pi^m\rho^n\topr[p_i](i,\rho_i^n)$, with
        $p_i=\tr(\ext{\pi_i}^\dagger\ext{\pi_i}\rho^n)$ and
        $\rho_i^n=\frac{\ext{\pi_i}\rho^n\ext{\pi_i}^\dagger}{p_i}$.
      \end{itemize}
    \item There exist $n$, $p_1,\dots,p_n$ and $t_1,\dots,t_n$ such that
      $t\topr[p_i]t_i$, in which case $\pi^mt\topr[p_i]t_i$;
    \item $t$ contains a free variable and does not rewrite, in which case the
      same is true for $\pi^mt$.
    \end{itemize}
  \item Let $\Gamma,\Delta,\vdash t\otimes r:n+m$ as a consequence of
    $\Gamma\vdash t:n$, $\Delta\vdash r:m$ and rule $\otimes$. Then, by the
    induction hypothesis, one of the following happens:
    \begin{itemize}
    \item There exist $k$, $p_1,\dots,p_n$, and $t_1,\dots,t_k$ such that
      $t\topr[p_i]t_i$, in which case $t\otimes r\topr[p_i]t_i\otimes r$.
    \item There exist $k$, $p_1,\dots,p_n$, and $r_1,\dots,r_k$ such that
      $r\topr[p_i]r_i$, in which case $t\otimes r\topr[p_i]t\otimes r_i$.
    \item $t$ is a value and $r$ does not rewrite. The only values which can be
      typed by $n$ are:
      \begin{itemize}
      \item $t=x$, then $t\otimes r$ contains a free variable and does not
        rewrite.
      \item $t=\bigotimes x_i$, then $t\otimes r$ contains free variables and
        does not rewrite.
      \item $t=\rho^n$, then:
        \begin{itemize}
        \item If $r=\rho^m$, $t\otimes r\topr \rho'$, with
          $\rho'=\rho^n\otimes\rho^m$.
        \item If $r$ contains a free variable and does not rewrite, then the
          same is true for $t\otimes r$.
        \end{itemize}
      \end{itemize}
    \item $t$ contains a free variable and does not rewrite, and $r$ does not
      rewrite, in which case $t\otimes r$ contains a free variable and does not
      rewrite.
    \end{itemize}
  \item Let $\Gamma\vdash(b^m,\rho^n):(m,n)$ as a consequence of rule
    $\mathsf{ax}_{\mathsf{am}}$. Then, we are done since $(b^m,\rho^n)$ is a
    value.
  \item Let $\Gamma\vdash\letcase xr{t_0,\dots,t_{2^m-1}}:A$ as a consequence of
    $x:n\vdash t_i:A$ for $i=0,\dots,2^m-1$, $\Gamma\vdash r:(m,n)$, and rule
    $\mathsf{lc}$. By the induction hypothesis, the possible cases for $r$ are:
    \begin{itemize}
    \item $r$ is a value. The possible values which can be typed by $(m,n)$ are:
      \begin{itemize}
      \item $r=y$, in which case $\letcase xr{t_0,\dots,t_{2^m-1}}$ contains a
        free variable and does not reduce.
      \item $r=(b^m,\rho^n)$, in which case $\letcase
        xr{t_0,\dots,t_{2^m-1}}\topr t_{b^m}[\rho^n/x]$.
      \end{itemize}
    \item There exist $k$, $p_1,\dots,p_k$, and $r_1,\dots,r_k$ such that
      $r\topr[p_i]r_i$, and $\letcase
      xr{t_0,\dots,t_{2^m-1}}\!\topr[p_i]\!\letcase
      x{r_i}{t_0,\dots,t_{2^m-1}}$.
    \item $r$ contains a free variable and does not rewrite, in which case the
      same is true for $\letcase xr{t_0,\dots,t_{2^m-1}}$.
    \end{itemize}
  \end{itemize}
\item We proceed by induction on the derivation of $\Gamma\Vdash t:A$.
  \begin{itemize}
  \item Let $\Gamma,x:A\Vdash x:A$ as a consequence of rule $\mathsf{ax}$. Then,
    we are done since $x$ is a free variable and does not rewrite.
  \item Let $\Gamma\Vdash\lambda x.t:A\multimap B$ as a consequence of
    $\Gamma,x:A\Vdash t:B$ and rule $\multimap_i$. Then, by the induction
    hypothesis, either
    \begin{itemize}
    \item $t$ is not a \qvalue\ and there exists $r$ such that $t\toD r$, which
      case $\lambda x.t\toD\lambda x.r$;
    \item $t$ is a \qvalue, in which case $\lambda x.t$ is a \qvalue\ as well;
      or
    \item $t$ contains a free variable, and $t$ does not rewrite, in which case
      the same happens to $\lambda x.t$.
    \end{itemize}
  \item Let $\Gamma,\Delta\Vdash tr:B$ as a consequence of $\Gamma\Vdash
    t:A\multimap B$, $\Delta\Vdash r:A$ and rule $\multimap_e$. Then, by the
    induction hypothesis, one of the following cases happens:
    \begin{itemize}
    \item There exists $t'$ such that $t\toD t'$, in which case $tr\toD t'r$.
    \item There exists $r'$ such that $r\toD r'$, in which case $tr\toD tr'$.
    \item $t$ is a \qvalue\ and $r$ does not rewrite. The only \qvalue[s]\ which
      can be typed by $A\multimap B$ are:
      \begin{itemize}
      \item $t=x$, in which case $xr$ contains a free variable and does not
        rewrite.
      \item $t=\lambda x.v$, in which case $(\lambda x.v)r\toD v[r/x]$.
      \item $t=\sum_i p_i t_i$, where $\Gamma\Vdash t_i:A\multimap B$. Then,
        $tr\toD\sum_i p_i (t_ir)$.
      \end{itemize}
    \item $t$ is not a \qvalue, contains a free variable, and does not rewrite,
      and $r$ does not rewrite, in which case, if $t$ is not a sum, $tr$
      contains a free variable and does not rewrite. If $t=\sum_ip_it_i$ is a
      sum, $tr\toD\sum_ip_i(t_ir)$.
    \end{itemize}
  \item Let $\Gamma\Vdash\rho^n:n$ as a consequence of rule $\mathsf{ax_\rho}$.
    Then, we are done since $\rho^n$ is a \qvalue.
  \item Let $\Gamma\Vdash U^mt:n$ as a consequence of $\Gamma\Vdash t:n$ and
    rule $\mathsf{u}$. Then, by the induction hypothesis, one of the following
    cases happens:
    \begin{itemize}
    \item $t$ is a \qvalue. Since $\lambda x.v$ cannot be typed by $n$, the only
      \qvalue[s]\ that can be typed by $n$ are either $\rho^n$, or they contain
      free variables:
      \begin{itemize}
      \item Let $t=\rho^n$, then $U^m\rho^n\toD\rho'$, with
        $\rho'=\ext{U^m}\rho^n\ext{U^m}^\dagger$.
      \item Let $t$ contain a free variable. Notice that it can only be either a
        free variable by itself, a tensor of \qvalue[s]\ containing free
        variables, or a linear combination of different \qvalue[s]\ containing
        free variables. In any case, $t$ contains a free variable and does not
        rewrite. Hence, $U^mt$ contains a free variable and does not rewrite.
      \end{itemize}
    \item There exists $r$ such that $t\toD r$, in which case $U^mt\toD r$;
    \item $t$ contains a free variable and does not rewrite, in which case the
      same is true for $U^mt$.
    \end{itemize}
  \item Let $\Gamma\Vdash \pi^mt:(m,n)$ as a consequence of $\Gamma\Vdash t:n$
    and rule $\mathsf{m}$. Then, by the induction hypothesis, one of the
    following cases happens:
    \begin{itemize}
    \item $t$ is a \qvalue. Since $\lambda x.v$ cannot be typed by $n$, the only
      \qvalue[s]\ that can be typed by $n$ are either $\rho^n$, or they contain
      free variables:
      \begin{itemize}
      \item Let $t=\rho^n$, then $\pi^m\rho^n\toD\rho'$, with
        $\rho'=\sum_i\ext{\pi_i}\rho^n\ext{\pi_i}^\dagger$.
      \item Let $t$ contain a free variable. Notice that it can only be either a
        free variable by itself, a tensor of \qvalue[s]\ containing free
        variables, or a linear combination of different \qvalue[s]\ containing
        free variables. In any case, $t$ contains a free variable and does not
        rewrite. Hence, $\pi^mt$ contains a free variable and does not rewrite.
      \end{itemize}
    \item There exists $r$ such that $t\toD r$, in which case $\pi^mt\toD r$;
    \item $t$ contains a free variable and does not rewrite, in which case the
      same is true for $\pi^mt$.
    \end{itemize}
  \item Let $\Gamma,\Delta,\Vdash t\otimes r:n+m$ as a consequence of
    $\Gamma\Vdash t:n$, $\Delta\Vdash r:m$ and rule $\otimes$. Then, by the
    induction hypothesis, one of the following happens:
    \begin{itemize}
    \item There exists $t'$ such that $t\toD t'$, in which case $t\otimes r\toD
      t'\otimes r$.
    \item There exists $r'$ such that $r\toD r'$, in which case $t\otimes r\toD
      t\otimes r'$.
    \item $t$ is a \qvalue\ and $r$ does not rewrite. The only \qvalue[s]\ that
      can be typed by $n$ are either $\rho^n$, or they contain free variables.
      \begin{itemize}
      \item Let $t=\rho^n$, then:
        \begin{itemize}
        \item If $r=\rho^m$, $t\otimes r\toD \rho'$, with
          $\rho'=\rho^n\otimes\rho^m$.
        \item If $r$ contains a free variable and does not rewrite, then the
          same is true for $t\otimes r$.
        \end{itemize}
      \item Let $t$ contain a free variable. Then $t\otimes r$ contains a free
        variable and does not rewrite.
      \end{itemize}
    \item $t$ contains a free variable and does not rewrite, and $r$ does not
      rewrite, in which case $t\otimes r$ contains a free variable and does not
      rewrite.
    \end{itemize}
  \item Let $\Gamma\Vdash\qletcase xr{t_0,\dots,t_{2^m-1}}:A$ as a consequence
    of $x:n\Vdash t_i:A$ for $i=0,\dots,2^m-1$, $\Gamma\Vdash r:(m,n)$, and rule
    $\mathsf{lc}$. By the induction hypothesis, the possible cases for $r$ are:
    \begin{itemize}
    \item $r$ is a \qvalue. The only possible \qvalue[s]\ that can be typed by
      $(m,n)$ are variables or linear combination of variables. In any case, we
      have that $\qletcase xr{t_0,\dots,t_{2^m-1}}$ contains at least a free
      variable and does not rewrite.
    \item There exists $r'$ such that $r\toD r'$, in which case, we have to
      distinguish two cases:
      \begin{itemize}
      \item $r\neq\pi^m\rho^n$, and therefore $\qletcase
        xr{t_0,\dots,t_{2^m-1}}\toD \qletcase x{r'}{t_0,\dots,t_{2^m-1}}$.
      \item $r=\pi^m\rho^n$, and therefore $\qletcase
        xr{t_0,\dots,t_{2^m-1}}\toD \sum_ip_it_i[\rho^n_i/x]$, where
        $p_i=\tr(\ext{\pi_i}^\dagger\ext{\pi_i}\rho^n)$ and
        $\rho_i^n=\frac{\ext{\pi_i}\rho^n\ext{\pi_i}^\dagger}{p_i}$.
      \end{itemize}
    \item $r$ contains a free variable and does not rewrite, in which case the
      same is true for $\qletcase xr{t_0,\dots,t_{2^m-1}}$.
    \end{itemize}
  \item Let $\Gamma\Vdash\sum_ip_it_i:A$ as a consequence of $\Gamma\Vdash
    t_i:A$, $\sum_ip_i=1$, and rule $+$. If $\sum_ip_it_i$ is a \qvalue, then we
    are done. If it is not a value, then one of the following cases is true:
    \begin{itemize}
    \item $t_j=t_k$ for some $j\neq k$, in which case $\sum_ip_it_i\toD
      (\sum_{i\neq j,k}p_it_i)+(p_j+p_k)t_j$.
    \item At least one $t_i$ is not a value. By the induction hypothesis, if
      $t_i$ is not a value, either it rewrites, or it contains a free variable
      and does not rewrite. If at least one $t_i$ rewrites, then $\sum_ip_it_i$
      rewrites. If none of these rewrites and at least one contains a free
      variable, then $\sum_ip_it_i$ does not rewrite and contain a free
      variable. \qed
    \end{itemize}
  \end{itemize}
\end{enumerate}

\end{document}